\documentclass{article}
\pdfoutput=1

\usepackage{algpseudocode}
\usepackage[Algorithm]{algorithm}
\usepackage{graphicx}
\usepackage{multirow}
\usepackage{caption}
\usepackage{amsmath}
\usepackage{amssymb}
\usepackage{algorithm}
\usepackage{subfig}
\usepackage{dblfloatfix}
\usepackage{amsthm}
\usepackage{url}

\newtheorem{theorem}{Theorem}[section]

\newtheorem{lemma}[theorem]{Lemma}
\newtheorem{corollary}[theorem]{Corollary}

\newtheorem{definition}{Definition}
\usepackage[margin=1.0in]{geometry}

\begin{document}
\title{Constant Time Updates in Hierarchical Heavy Hitters}
\author{
	Ran Ben Basat\\
	Technion\\
	\texttt{sran@cs.technion.ac.il}
	\and
	Gil Einziger\\
	Nokia Bell Labs\\
	\texttt{gil.einziger@nokia}
	\and
	Roy Friedman\\
	Technion\\
	\texttt{roy@cs.technion.ac.il}	
	\and
	Marcelo C. Luizelli\\
	UFRGS\\
	\texttt{mcluizelli@inf.ufrgs.br}	
	\and
	Erez Waisbard\\	
	Nokia Bell Labs\\
	\texttt{erez.waisbard@nokia.com}	
}

%
%
%
%


\date{}

\newcommand{\eps}{\epsilon}
\newcommand{\set}[1]{\left\{#1\right\}}
\newcommand{\ceil}[1]{ \left\lceil{#1}\right\rceil}
\newcommand{\floor}[1]{ \left\lfloor{#1}\right\rfloor}
\newcommand{\logp}[1]{\log\parentheses{#1}}
\newcommand{\clog}[1]{ \ceil{\log{#1}}}
\newcommand{\clogp}[1]{ \ceil{\logp{#1}} }
\newcommand{\flog}[1]{ \floor{\log{#1}}}
\newcommand{\parentheses}[1]{ \left({#1}\right)}
\newcommand{\abs}[1]{ \left|{#1}\right|}

\newcommand{\cdotpa}[1]{\cdot\parentheses{#1}}
\newcommand{\inc}[1]{$#1 = #1 + 1$}
\newcommand{\dec}[1]{$#1 = #1 - 1$}
\newcommand{\range}[2][0]{#1,1,\ldots,#2}
\newcommand{\frange}[1]{\set{\range{#1}}}
\newcommand{\xrange}[1]{\frange{#1-1}}
\newcommand{\oneOverE}{ \frac{1}{\eps} }
\newcommand{\oneOverT}{ \frac{1}{\tau} }
\newcommand{\smallMultError}{(1+o(1))}
\newcommand{\lowerbound}{\max \set{\log W ,\frac{1}{2\epsilon+W^{-1}}}}
\newcommand{\smallEpsLowerbound}{\window\logp{\frac{1}{\weps}}}
\newcommand{\smallEpsMemoryTheta}{$\Theta\parentheses{\smallEpsMemoryConsumption}$}
\newcommand{\smallEpsMemoryConsumption}{W\cdot\logp{\frac{1}{\weps}}}

\newcommand{\largeEpsRestriction}{For any \largeEps{},}
\newcommand{\largeEps}{$\eps^{-1} \le 2W\left(1-\frac{1}{\logw}\right)$}
\newcommand{\smallEpsRestriction}{For any \smallEps{},}
\newcommand{\smallEps}{$\eps^{-1}>2W\left(1-\frac{1}{\logw}\right)=2\window(1-o(1))$}
\newcommand{\bc}{{\sc Basic-Counting}}
\newcommand{\bs}{{\sc Basic-Summing}}
\newcommand{\windowcounting}{ {\sc $(W,\epsilon)$-Window-Counting}}

\newcommand{\query}[1][] {{\sc Query}$(#1)$}
\newcommand{\add}  [1][] {{\sc Add}$(#1)$}

\newcommand{\window}{W}
\newcommand{\logw}{\log \window}
\newcommand{\flogw}{\floor{\log \window}}
\newcommand{\weps}{\window\epsilon}
\newcommand{\wt}{\window\tau}
\newcommand{\logweps}{\logp{\weps}}
\newcommand{\logwt}{\logp{\wt}}
\newcommand{\bitarray}{b}
\newcommand{\currentBlockIndex}{i}
\newcommand{\currentBlock}{\bitarray_{\currentBlockIndex}}
\newcommand{\remainder}{y}
\newcommand{\numBlocks}{k}
\newcommand{\sumOfBits}{B}
\newcommand{\blockSize}{\frac{\window}{\numBlocks}}
\newcommand{\iblockSize}{\frac{\numBlocks}{\window}}
\newcommand{\threshold}{\blockSize}
\newcommand{\halfBlock}{\frac{\window}{2\numBlocks}}
\newcommand{\blockOffset}{m}
\newcommand{\inputVariable}{x}

\newcommand{\bcTableColumnWhh}{1.5cm}
\newcommand{\bsTableColumnWidth}{1.7cm}
\newcommand{\bsExtendedTableColumnWidth}{3cm}
\newcommand{\bcExtendedTableColumnWidth}{2.8cm}
\newcommand{\bcNarrowTableColumnWidth}{1.5cm}
\newcommand{\bsNarrowTableColumnWidth}{1.5cm}
\newcommand{\bsWorstCaseTableColumnWidth}{2cm}

\newcommand{\bsrange}{ R }
\newcommand{\bsReminderPercisionParameter}{ \gamma }
\newcommand{\bsest}{ \widehat{\bssum}}
\newcommand{\bssum}{ S^W }
\newcommand{\bsFracInput}{ \inputVariable' }
\newcommand{\bserror}{ \bsrange\window\epsilon }
\newcommand{\bsfractionbits}{ \frac{\bsReminderPercisionParameter}{\epsilon} }
\newcommand{\bsReminderFractionBits}{ \upsilon}
\newcommand{\bsAnalysisTarget}{ \bssum}
\newcommand{\bsAnalysisEstimator}{ \widehat{\bsAnalysisTarget}}
\newcommand{\bsAnalysisError}{ \bsAnalysisEstimator - \bsAnalysisTarget}
\newcommand{\bsRoundingError}{ \xi}


\newcommand{\neps}{\ensuremath{\winSize\eps}}
\newcommand{\Neps}{\ensuremath{\maxWinSize\eps}}
\newcommand{\logn}{\ensuremath{\log\winSize}}
\newcommand{\logN}{\ensuremath{\log\maxWinSize}}
\newcommand{\logneps}{\ensuremath{\logp\neps}}
\newcommand{\logNeps}{\ensuremath{\logp\Neps}}
\newcommand{\oneOverEps}{\ensuremath{\frac{1}{\eps}}}
\newcommand{\winSize}{\ensuremath{n}}
\newcommand{\maxWinSize}{\ensuremath{N}}
\newcommand{\curTime}{\ensuremath{t}}
\newcommand\Tau{\mathrm{T}}
\newcommand{\offset}{\ensuremath{\mathit{offset}}}
\newcommand{\roundedOOE}{k}
\newcommand{\numLargeBlocks}{\frac{\roundedOOE}{4}}
\newcommand{\numSmallBlocks}{\frac{\roundedOOE}{2}}

\newcommand{\remove}{{\sc Remove()}}
\newcommand{\merge}[1]{{\sc Merge(#1)}}
\newcommand{\counting}{{\sc Counting}}
\newcommand{\summing}{{\sc Summing}}
\newcommand{\freq}{{\sc Frequency Estimation}}
\newtheorem{problem}[theorem]{Problem}

\newenvironment{remark}[1][Remark]{\begin{trivlist}
		\item[\hskip \labelsep {\bfseries #1}]}{\end{trivlist}}

\newcommand{\NB}{\psi}
\newcommand{\NBound}{{Z_{1 - \frac{{{\delta _s}}}{2}}}V{\varepsilon_s}^{ - 2}}
\newcommand{\matrixCellWidth}{5.8cm}
\renewcommand{\arraystretch}{1.33}
\newcommand{\cmark}{\ding{51}}%
\newcommand{\xmark}{\ding{55}}%
%
%
%
%
%

\maketitle
\begin{abstract}
Monitoring tasks, such as anomaly and DDoS detection, require identifying frequent flow aggregates based on common IP prefixes.
These are known as \emph{hierarchical heavy hitters} (HHH), where the hierarchy is determined based on the type of prefixes of interest in a given application.
The per packet complexity of existing HHH algorithms is proportional to the size of the hierarchy, imposing significant overheads.
	
In this paper, we propose a randomized constant time algorithm for HHH.
We prove probabilistic precision bounds backed by an empirical evaluation.
Using four real Internet packet traces, we demonstrate that our algorithm indeed obtains comparable accuracy and recall as previous works, while running up to 62 times faster.
Finally, we extended Open vSwitch (OVS) with our algorithm and showed it is able to handle 13.8 million packets per second. 
In contrast, incorporating previous works in OVS only obtained 2.5 times lower throughput.
	
\end{abstract}

\section{Introduction}
Network measurements are essential for a variety of network functionalities such as traffic engineering, load balancing, quality of service, caching, anomaly and intrusion detection~\cite{LBSigComm,DevoFlow,ApproximateFairness,TrafficEngeneering,IntrusionDetection2,7218487,CONGA,TinyLFU}.
A major challenge in performing and maintaining network measurements comes from rapid line rates and the large number of active flows.

Previous works suggested identifying \emph{Heavy Hitter} (HH) flows~\cite{Woodruff16} that account for a large portion of the traffic.
Indeed, approximate HH are used in many functionalities and can be captured quickly and efficiently~\cite{HashPipe,ICCCNPaper,WCSS,DIM-SUM,RAP}.
However, applications such as anomaly detection and \emph{Distributed Denial of Service} (DDoS) attack detection require more sophisticated measurements~\cite{Zhang:2004:OIH:1028788.1028802,Sekar2006}.
In such attacks, each device generates a small portion of the traffic but their combined volume is overwhelming.
HH measurement is therefore insufficient as each individual device is not a heavy hitter.
\begin{figure}[t!]
	\includegraphics[width = 0.8\columnwidth]{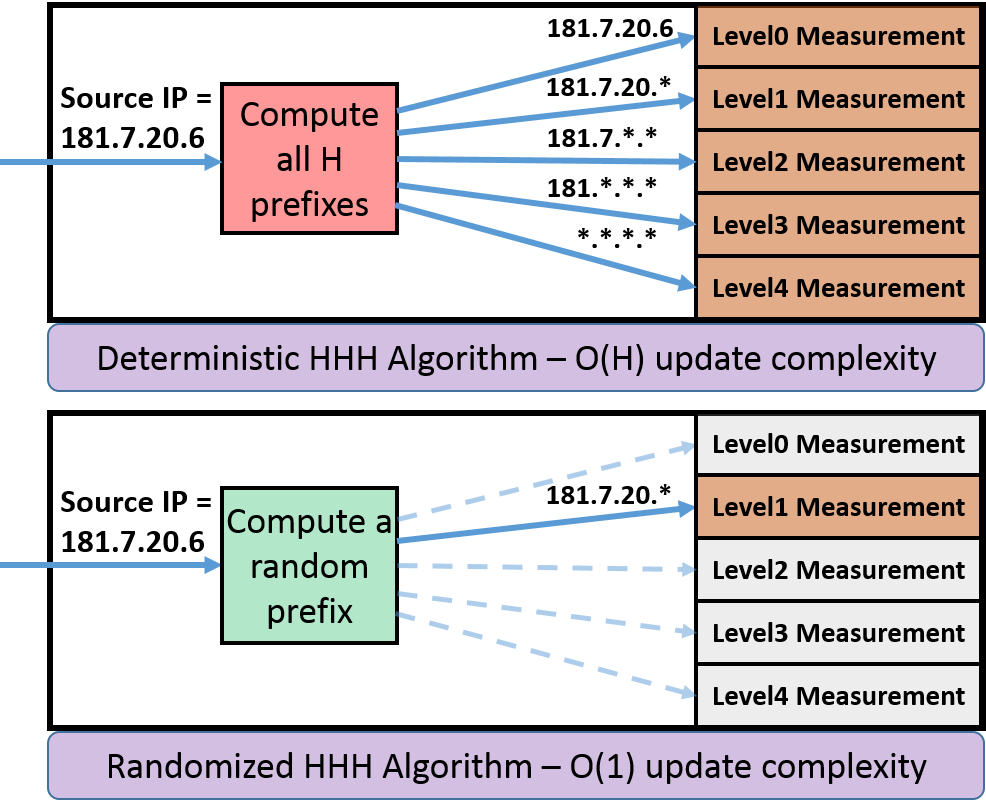}
	
	\caption{A high level overview of this work. Previous algorithms' update requires $\Omega(H)$ run time, while we perform at most a single $O(1)$ update.
		}
	\label{fig:contribution}
\end{figure}

\emph{Hierarchical Heavy Hitters} (HHH) account aggregates of flows that share certain IP prefixes.
The structure of IP addresses implies a prefix based hierarchy as defined more precisely below.
In the DDoS example, HHH can identify IP prefixes that are suddenly responsible for a large portion of traffic and such an anomaly may very well be a manifesting attack.
Further, HHH can be collected in one dimension, e.g., a single source IP prefix hierarchy, or in multiple dimensions, e.g., a hierarchy based on both source and destination IP~prefixes. 

Previous works~\cite{HHHMitzenmacher,CormodeHHH} suggested deterministic algorithms whose update complexity is proportional to the hierarchy's size.
These algorithms are currently too slow to cope with line speeds.
For example, a 100 Gbit link may deliver over 10 million packets per second, but previous HHH algorithms cannot cope with this line speed on existing hardware.
The transition to IPv6 is expected to increase hierarchies' sizes and render existing approaches even~slower.

Emerging networking trends such as \emph{Network Function Virtualization} (NFV) enable virtual deployment of network functionalities.
These are run on top of commodity servers rather than on custom made hardware, thereby improving the network's flexibility and reducing operation costs.
These trends further motivate fast software based measurement algorithms.

\subsection{Contributions}

First, we define a probabilistic relaxation of the HHH problem.
Second, we introduce \emph{Randomized HHH} (a.k.a. RHHH), a novel randomized algorithm that solves probabilistic HHH over single and multi dimensional hierarchical domains.
Third, we evaluate RHHH on four different real Internet traces and demonstrate a speedup of up to X62 while delivering similar accuracy and recall ratios.
Fourth, we integrate RHHH with \emph{Open vSwitch} (OVS) and demonstrate a capability of monitoring HHH at line speed, achieving a throughput of up to 13.8M packets per second.
Our algorithm also achieves X2.5 better throughput than previous approaches.
To the best of our knowledge, our work is the first to perform OVS multi dimensional HHH analysis in line speed.

Intuitively, our RHHH algorithm operates in the following way, as illustrated in Figure~\ref{fig:contribution}:
We maintain an instance of a heavy-hitters detection algorithm for each level in the hierarchy, as is done in~\cite{HHHMitzenmacher}.
However, whenever a packet arrives, we randomly select only a single level to update using its respective instance of heavy-hitters rather than updating all levels (as was done in~\cite{HHHMitzenmacher}).
Since the update time of each individual level is $O(1)$, we obtain an $O(1)$ \emph{worst case} update time.
The main challenges that we address in this paper are in formally analyzing the accuracy of this scheme and exploring how well it works in practice with a concrete implementation.

The update time of previous approaches is $O(H)$, where $H$ is the size of the hierarchy.
An alternative idea could have been to simply sample each packet with probability $\frac{1}{H}$, and feed the sampled packets to previous solutions.
However, such a solution only provides an $O(1)$ \emph{amortized} running time.
Bounding the worst case behavior to $O(1)$ is important when the counters are updated inside the data path.
In such cases, performing an occasional very long operation could both delay the corresponding ``victim'' packet, and possibly cause buffers to overflow during the relevant long processing.
Even in off-path processing, such as in an NFV setting, occasional very long processing creates an unbalanced workload, challenging schedulers and resource allocation schemes.

\paragraph*{Roadmap} The rest of this paper is organized as follows:
We survey related work on HHH in Section~\ref{sec:related}.
We introduce the problem and our probabilistic algorithm in Section~\ref{sec:randomized}.
For presentational reasons, we immediately move on to the performance evaluation in Section~\ref{sec:eval} followed by describing the implementation in OVS in Section~\ref{sec:ovs}.
We then prove our algorithm and analyze its formal guarantees in Section~\ref{sec:analysis}.
Finally, we conclude with a discussion in Section~\ref{sec:discussion}.

 \begin{table*}[t]
	 \footnotesize
 	\addtolength{\tabcolsep}{-0pt} \addtolength{\parskip}{-0.1mm} \center{
 		\begin{tabular}{|c|c|c|c|c|c|}
 			\hline
 			\textbf{Src/Dest } & \textbf{{*}} & \textbf{d1.{*}} & \textbf{d1.d2.{*}} & \textbf{d1.d2.d3.{*}} & \textbf{d1.d2.d3.d4}\tabularnewline
 			\hline
 			\hline
 			\textbf{{*}} & ({*},{*}) & ({*},d1.{*}) & ({*},d1.d2.{*}) & ({*},d1.d2.d3.{*}) & ({*},d1.d2.d3.d4)\tabularnewline
 			\hline
 			\textbf{s1.{*}} & (s1.{*},{*}) & (s1.{*},d1.{*}) & (s1.{*},d1.d2.{*}) & (s1.{*},d1.d2.d3.{*}) & (s1.{*},d1.d2.d3.d4)\tabularnewline
 			\hline
 			\textbf{s1.s2.{*}} & (s1.s2.{*},{*}) & (s1.s2.{*},d1.{*}) & (s1.s2.{*},d1.d2.{*}) & (s1.s2.{*},d1.d2.d3.{*}) & (s1.s2.{*},d1.d2.d3.d4)\tabularnewline
 			\hline
 			\textbf{s1.s2.s3.{*}} & (s1.s2.s3.{*},{*}) & (s1.s2.s3.{*},d1.{*}) & (s1.s2.s3.{*},d1.d2.{*}) & (s1.s2.s3.{*},d1.d2.d3.{*}) & (s1.s2.s3.{*},d1.d2.d3.d4)\tabularnewline
 			\hline
 			\textbf{s1.s2.s3.s4} & (s1.s2.s3.s4,{*}) & (s1.s2.s3.s4,d1.{*}) & (s1.s2.s3.s4,d1.d2.{*}) & (s1.s2.s3.s4,d1.d2.d3.{*}) & (s1.s2.s3.s4,d1.d2.d3.d4)\tabularnewline
 			\hline
 		\end{tabular}
 		
 		{{\caption{An example of the lattice induced by a two dimensional source/destination byte hierarchy. The top left corner (*,*) is fully general while the bottom right (s1.s2,s3.s4,d1.d2.d3.d4) is fully specified. The parents of each node are directly above it and directly to the left.}
 				\label{tbl:example}
 			}}}
 		\end{table*}
\normalsize

	

\section{Related Work}
\label{sec:related}

In one dimension, HHH were first defined by~\cite{Cormode2003}, which also introduced the first streaming algorithm to approximate them.
Additionally,~\cite{HHHSwitch} offered a TCAM approximate HHH algorithm for one dimension.
The HHH problem was also extended to multiple dimensions~\cite{Cormode2004,CormodeHHH,Hershberger2005,Zhang:2004:OIH:1028788.1028802,HHHMitzenmacher}.

The work of~\cite{Lin2007} introduced a single dimension algorithm that requires \small$O\left(\frac{H^2}{\epsilon}\right)$\normalsize space, where the symbol $H$ denotes the size of the hierarchy and $\epsilon$ is the allowed relative estimation error for each single flow's frequency.
Later,~\cite{Truong2009} introduced a two dimensions algorithm that requires \small $O\left(\frac{H^{3/2}}{\epsilon}\right)$\normalsize space and update time\footnote{Notice that in two dimensions, $H$ is a square of its counter-part in one dimension.}.
In~\cite{CormodeHHH}, the trie based Full Ancestry and Partial Ancestry algorithms were proposed. These use $O\left(\frac{H\log(N\epsilon)}{\epsilon}\right)$ space and requires $O\left(H\log(N\epsilon)\right)$ time per update.

The seminal work of \cite{HHHMitzenmacher}~introduced and evaluated a simple multi dimensional HHH algorithm.
Their algorithm uses a separate copy of Space Saving~\cite{SpaceSavings} for each lattice node and upon packet arrival, all lattice nodes are updated.  
Intuitively, the problem of finding hierarchical heavy hitters can be reduced to solving multiple non hierarchical heavy hitters problems, one for each possible query.
This algorithm provides strong error and space guarantees and its update time does not depend on the stream length.
Their algorithm requires $O\left(\frac{H}{\epsilon}\right)$ space and its update time for unitary inputs is $O\left(H\right)$ while for weighted inputs it is $O\left(H \log \frac{1}{\epsilon}\right)$.

The update time of existing methods is too slow to cope with modern line speeds and the problem escalates in NFV environments that require efficient software implementations.
This limitation is both empirical and asymptotic as some settings require large hierarchies.  

Our paper describes a novel algorithm that solves a probabilistic version of the hierarchical heavy hitters problem.
We argue that in practice, our solution's quality is similar to previously suggested deterministic approaches while the runtime is dramatically improved.
Formally, we improve the update time to $O(1)$, but require a minimal number of packets to provide accuracy guarantees.
We argue that this trade off is attractive for many modern networks that route a continuously increasing number of packets.


\section{Randomized HHH (RHHH)}
\label{sec:randomized}
We start with an intuitive introductory to the field as well as preliminary definitions and notations.
Table~\ref{tbl:notations} summarizes notations used in this work.
\subsection{Basic terminology}
\label{sec:terminology}

We consider IP addresses to form a hierarchical domain with either bit or byte size granularity.
\emph{Fully specified} IP addresses are the lowest level of the hierarchy and can be generalized.
We use $\mathcal U$ to denote the domain of fully specified items.
For example, $181.7.20.6$ is a fully specified IP address and $181.7.20.*$ generalizes it by a single byte.
Similarly, $181.7.*$ generalizes it by two bytes and formally, a fully specified IP address is generalized by any of its prefixes.
The \emph{parent} of an item is the longest prefix that generalizes it.

In two dimensions, we consider a tuple containing source and destination IP addresses.
A fully specified item is fully specified in both dimensions.
For example,
$(\langle181.7.20.6\rangle\to \langle 208.67.222.222\rangle)$
is fully specified.
In two dimensional hierarchies, each item has two parents, e.g.,
$(\langle181.7.20.*\rangle\to \langle 208.67.222.222\rangle)$
and
$(\langle181.7.20.6\rangle\to \langle 208.67.222.*\rangle)$
are both parents to\\
$(\langle181.7.20.6\rangle\to \langle 208.67.222.222\rangle)$.

\begin{definition}[Generalization]
	For two prefixes $p,q$, we denote $p \preceq q$ if in any dimension it is either a prefix of $q$ or is equal to $q$.
	We also denote the set of elements that are generalized by $p$ with $H_p\triangleq \{e\in \mathcal U\mid e\preceq p\}$, and those generalized by a set of prefixes $P$ by $H_P\triangleq \cup_{p\in P} H_p$. If $p \preceq q$ and $p\neq q$, we denote $p\prec q$.
\end{definition}
In a single dimension, the generalization relation defines a vector going from fully generalized to fully specified.
In two dimensions, the relation defines a lattice where each item has two parents.
A byte granularity two dimensional lattice is illustrated in Table~\ref{tbl:example}.
In the table, each lattice node is generalized by all nodes that are upper or more to the left.
The most generalized node $(*,*)$ is called \emph{fully general} and the most specified node $(s1.s2.s3.s4, d1.d2.d3.d4)$ is called \emph{fully specified}.
We denote $H$ the hierarchy's size as the number of nodes in the lattice.
For example, in IPv4, byte level one dimensional hierarchies imply $H=5$ as each IP address is divided into four bytes and we also allow querying $*$.

\begin{definition}
	Given a prefix $p$ and a set of prefixes $P$,
	we define $G(p|P)$ as the set of prefixes:
	$$\left\{ {h:h \in P,h \prec p,\nexists\,h' \in P\,\,s.t.\,h \prec h' \prec p} \right\}.$$
\end{definition}
Intuitively, $G(p|P)$ are the prefixes in $P$ that are most closely generalized by $p$.
E.g., let $p=<142.14.*>$ and the set \\$P = \left\{ {<142.14.13.*>,<142.14.13.14>} \right\}$, then $G(p|P)$ only contains $<142.14.13.*>$.

We consider a stream $\mathbb{S}$, where at each step a packet of an item $e$ arrives.
Packets belong to a hierarchical domain of size $H$, and can be generalized by multiple prefixes as explained above. 
Given a fully specified item $e$, $f_e$ is the number of occurrences $e$ has in $\mathbb{S}$.
Definition~\ref{def:frequency} extends this notion to prefixes.
\begin{definition} (Frequency)
	\label{def:frequency}
	Given a prefix $p$, the frequency of $p$ is: $$f_p \triangleq \sum\nolimits_{e \in H_p} {f_e} .$$
\end{definition}




Our implementation utilizes Space Saving~\cite{SpaceSavings}, a popular (non hierarchical) heavy hitters algorithm, but other algorithms can also be used.
Specifically, we can use any \emph{counter algorithm} that satisfies Definition~\ref{Def:probFE} below and can also find heavy hitters, such as~\cite{frequent4,BatchDecrement,LC}.
We use Space Saving because it is believed to have an empirical edge over other algorithms~\cite{SpaceSavingIsTheBest,SpaceSavingIsTheBest2009,SpaceSavingIsTheBest2010}.

\begin{table}[h]
	\centering

	\begin{tabular}{|c|l|}
		
		\hline
		Symbol & Meaning \tabularnewline
		\hline
		$\mathbb{S}$ & Stream \tabularnewline
		\hline
		$N$ & Current number of packets (in all flows)  \tabularnewline
		\hline
		$H$ & Size of Hierarchy  \tabularnewline
		\hline
		$V$ & Performance parameter, $V\ge H$  \tabularnewline
		\hline
		$S^i_x$ & Variable for the i'th appearance of a prefix $x$. 
		\tabularnewline
		\hline
		$S_{x}$ & Sampled prefixes with id $x$. 
		\tabularnewline
		\hline
		$S$ &  Sampled prefixes from all ids. 
		\tabularnewline
		
		\hline
		$\mathcal U$ & Domain of fully specified items. \tabularnewline
		\hline
		$\epsilon,\epsilon_s,\epsilon_a $ & Overall, sample, algorithm's error guarantee.  \tabularnewline
		\hline
		$\delta,\delta_s , \delta_a$ &   Overall, sample, algorithm confidence. \tabularnewline
		\hline
		$\theta$ & Threshold parameter.    \tabularnewline
		\hline

		$C_{q|P}$ & Conditioned frequency of $q$ with respect to  $P$ \tabularnewline
		\hline
		$G(q|P)$ & Subset of $P$ with the closest prefixes to q.  \tabularnewline
		\hline
		$f_q$ & Frequency of prefix q \tabularnewline
		\hline
		$\widehat{f^{+}_q},\widehat{f^{-}_q}$ & Upper,lower bound for $f_q$  \tabularnewline
		\hline
	\end{tabular}
	\caption{List of Symbols}
	\label{tbl:notations}
	
\end{table}
\normalsize




The minimal requirements from an algorithm to be applicable to our work are defined in Definition~\ref{Def:probFE}.
This is a weak definition and most counter algorithms satisfy it with $\delta =0$.
Sketches~\cite{CountSketch,CMSketch,TinyTable} can also be applicable here, but to use them, each sketch should also maintain a list of heavy hitter items (Definition~\ref{def:HH}).
\begin{definition}
\label{Def:probFE}
An algorithm solves the {\sc {$(\epsilon, \delta)$ - Frequency Estimation}} problem if for any prefix ($x$),
it provides  $\widehat{f_{x}}$ s.t.:
$$\Pr \left[ {\left| {{f_x} - \widehat {{f_x}}} \right| \le \varepsilon N} \right] \ge 1 - \delta . $$
\end{definition}

\begin{definition}[Heavy hitter (HH)]\label{def:HH}
Given a threshold $(\theta)$, a fully specified item $(e)$ is a \textbf{heavy hitter} if its frequency $(f_e)$ is above the threshold:  $\theta \cdot N$, i.e., $f_e \ge \theta\cdot N$.
\end{definition}

Our goal is to identify the hierarchical heavy hitter prefixes whose frequency is above the threshold $(\theta \cdot N)$.
However, if the frequency of a prefix exceeds the threshold then so is the frequency of all its ancestors.
For compactness, we are interested in prefixes whose frequency is above the threshold due to non HHH siblings.
This motivates the definition of  \emph{conditioned frequency} ($C_{p|P}$).
Intuitively, $C_{p|P}$ measures the \textbf{additional} traffic prefix $p$ adds to a set of previously selected HHHs ($P$), and it is defined as follows.
%
\begin{definition} (Conditioned frequency)
	\label{def:cf}
	The conditioned frequency of a prefix $p$ with respect to a prefix set $P$ is: $$C_{p\mid P}\triangleq\allowbreak \sum_{e\in H_{(P\cup\{p\})} \setminus H_P} f_e.$$
\end{definition}
$C_{p\mid P}$ is derived by subtracting the frequency of fully specified items that are already generalized by items in $P$ from $p$'s frequency ($f_p$).
In two dimensions, exclusion inclusion principles are used to avoid double counting.

We now continue and describe how exact hierarchical heavy hitters (with respect to $C_{p\mid P}$) are found.
To that end, partition the hierarchy to levels as explained in Definition~\ref{Def:L}.

\begin{definition}[Hierarchy Depth]
	\label{Def:L}
	Define $L$, the \emph{depth of a hierarchy}, as follows:
    Given a fully specified element $e$, we consider a set of prefixes such that:  $e \prec p_1 \prec p_2,..\prec p_L$ where $e \neq p_1 \neq p_2 \neq ... \neq p_{L}$ and $L$ is the maximal size of that set.
	We also define the function $level(p)$ that given a prefix $p$ returns $p$'s maximal location in the chain, i.e., the maximal chain of generalizations that ends in $p$.
\end{definition}

To calculate exact heavy hitters, we go over fully specified items ($level  0$) and add their heavy hitters to the set $HHH_0$.
Using $HHH_0$, we calculate conditioned frequency for prefixes in $level 1$ and if $C_{p|{HHH_0}} \ge \theta \cdot N$ we add $p$ to $HHH_1$.
We continue this process until the last level ($L$) and the exact heavy hitters are the set $HHH_L$. Next, we define $HHH$ formally.

\begin{definition}[Hierarchical HH (HHH)]
	\label{def:HHH}
	
		 The set $HHH_0$ contains the fully specified items $e$ s.t. $f_e \ge \theta\cdot N$.
		 Given a prefix $p$ from level($l$), $0\le l\le L$, we define:
		\[\begin{array}{l} HH{H_l} = 
		HH{H_{l -1}} \cup \left\{ {p:\left( {p \in level\left( l \right) \wedge {C_{p|HH{H_{l - 1}}}} \ge \theta  \cdot N} \right)} \right\} .
		\end{array}\]
	 The set of exact hierarchical heavy hitters $HHH$ is defined as the set $HHH_L$.
\end{definition}

For example, consider the case where $\theta N =100$ and assume that the following prefixes with their frequencies are the only ones above $\theta N$.
$p_1 = (<101.*>, 108)$ and $p_2 = (<101.102.*>, 102)$.
Clearly, both prefixes are heavy hitters according to Definition~\ref{def:HH}.
However, the conditioned frequency of $p1$ is $108-102  = 6$ and that of $p_2$ is 102.
Thus only $p_2$ is an HHH prefix.

Finding exact hierarchical heavy hitters requires plenty of space.
Indeed, even finding exact (non hierarchical) heavy hitters requires linear space~\cite{TCS-002}.
Such a memory requirement is prohibitively expensive and motivates finding approximate HHHs.  

\begin{definition}[$(\epsilon,\theta)-$approximate HHH]
\label{def:approxHHH}
An algorithm solves {\sc {$(\epsilon, \theta)$ - Approximate Hierarchical Heavy Hitters}} if after processing any stream $\mathbb{S}$ of length $N$, it returns a set of prefixes ($P$) that satisfies the following conditions:
	\begin{itemize}
		\item \textbf{Accuracy:} for every prefix $p\in P$,  ${\left| {{f_p} - \widehat {{f_p}}} \right| \le \varepsilon N}$.
		\item \textbf{Coverage:} for every prefix $q \notin P$:  ${{C_{q|P}} < \theta N}$.
	\end{itemize}
\end{definition}
Approximate HHH are a set of prefixes ($P$) that satisfies accuracy and coverage; there are many possible sets that satisfy both these properties.
Unlike exact HHH, we do no require that for $p \in P$, $C_{p|P} \ge \theta N$.
Unfortunately, if we add such a requirement then~\cite{Hershberger2005} proved a lower bound of $\Omega\left(\frac{1}{\theta^{d+1}}\right)$ space, where $d$ is the number of dimensions.
This is considerably more space than is used in our work ($\frac{H}{\epsilon}$) that when $\theta \propto \epsilon$ is also $\frac{H}{\theta}$.

Finally, Definition~\ref{def:deltaapproxHHH} defines the probabilistic approximate HHH problem that is solved in this paper.
\begin{definition}[$(\delta,\epsilon,\theta)-$approximate HHHs]
	\label{def:deltaapproxHHH}
	 An algorithm $\mathbb A$ solves {\sc {$(\delta,\epsilon, \theta)$ - Approximate Hierarchical Heavy Hitters}} if after processing any stream $\mathbb{S}$ of length $N$, it returns a set of prefixes $P$ that, for an arbitrary run of the algorithm, satisfies the following:
 \begin{itemize}
\item \textbf{Accuracy:} for every prefix $p\in P$, $$\Pr \left( {\left| {{f_p} - \widehat {{f_p}}} \right| \le \varepsilon N} \right) \ge 1 - \delta.$$
\item \textbf{Coverage:} given a prefix $q \notin P$,  $$\Pr \left( {{C_{q|P}} < \theta N} \right) \ge 1 - \delta .$$
\end{itemize}

\end{definition}
Notice that this is a simple probabilistic relaxation of Definition~\ref{def:approxHHH}.
Our next step is to show how it enables the development of faster algorithms.

\begin{algorithm}[h]
	\begin{algorithmic}[1]
\Statex Initialization: $\forall d\in[L]: HH[d] =$ HH\_Alg $(\epsilon_a^{-1})$
		\Function{Update}{ $x$}
		\State $d = randomInt(0,V)$
		\If {$d<H$}
		\State Prefix $p = x\&HH[d].mask$ \Comment{Bitwise AND}
		\State $HH[d].INCREMENT(p)$
		\EndIf
		\EndFunction

		\Function{Output}{$\theta$}
		\State $P = \phi$
		\For{Level $l = |H|$ down to $0$. }
		\For{ each $p$ in level $l$}
		\State \label{line:cp}$\widehat{C_{p|P}} = \widehat{f_p}^{+} + calcPred(p,P) $
		\State \label{line:accSample}$\widehat{C_{p|P}} = \widehat{C_{p|P}}+ 2{{Z_{1 - {\delta}}}\sqrt {NV} }$
		\If {$\widehat{C_{p|P}}\ge \theta N$}
		\State $ P = P \cup \{p\}$ \Comment{$p$ is an HHH candidate}
		\State $print\left(p, \widehat{f_p}^{-}, \widehat{f_p}^{+}\right)$
		\EndIf
		\EndFor
		\EndFor
		\State\Return $P$
		\EndFunction
		
	\end{algorithmic}
\normalsize
	\caption{Randomized HHH algorithm}
	\label{alg:Skipper}
\end{algorithm}

\begin{algorithm}[h]
	\begin{algorithmic}[1]
		\Function{calcPred}{prefix $p$, set $P$}
		\State $R = 0$
		\For{ each $h\in G(p|P)$}
		\State \label{alg:first}$R = R - \widehat{f_h}^{-}$
		\EndFor
		\State \Return $R$
		\EndFunction
	\end{algorithmic}
	\normalsize
	\caption{calcPred for one dimension }
	\label{alg:randHHH}
\end{algorithm}

\begin{algorithm}[h]
	\begin{algorithmic}[1]
		\Function{calcPred}{prefix $p$, set $P$}
		\State $R = 0$
		\For{ each $h\in G(p|P)$}
		\State \label{alg:second}$R = R - \widehat{f_h}^{-}$
		\EndFor
		\For{ each pair $h,h'\in G(p|P)$}
		\State $q=glb(h,h')$
		\If {$\not\exists h_3 \neq h,h'\in G(p|P), q \preceq h_3$}
		\State \label{alg:third}$R = R + \widehat{f_q}^{+}$
		\EndIf
		\EndFor
		\State \Return $R$
		\EndFunction
	\end{algorithmic}
	\normalsize
	\caption{calcPred for two dimensions}
	\label{alg:randHHH2D}
\end{algorithm}

\subsection{Randomized HHH}
Our work employs the data structures of~\cite{HHHMitzenmacher}.
That is, we use a matrix of $H$ independent HH algorithms, and each node is responsible for a single prefix pattern.

\newcommand{\fourFigWidth}{4.2cm}
\newcommand{\thirdFigWidth}{5.47cm}
\newcommand{\twoFigWidth}{8.2cm}
\begin{figure*}[t]
	\captionsetup{justification=centering}
	\begin{tabular}{cccc}
		\subfloat[Chicago15\label{ACC:CH15:2Dbyte} - 2D Bytes]{\includegraphics[width = \fourFigWidth ]
			{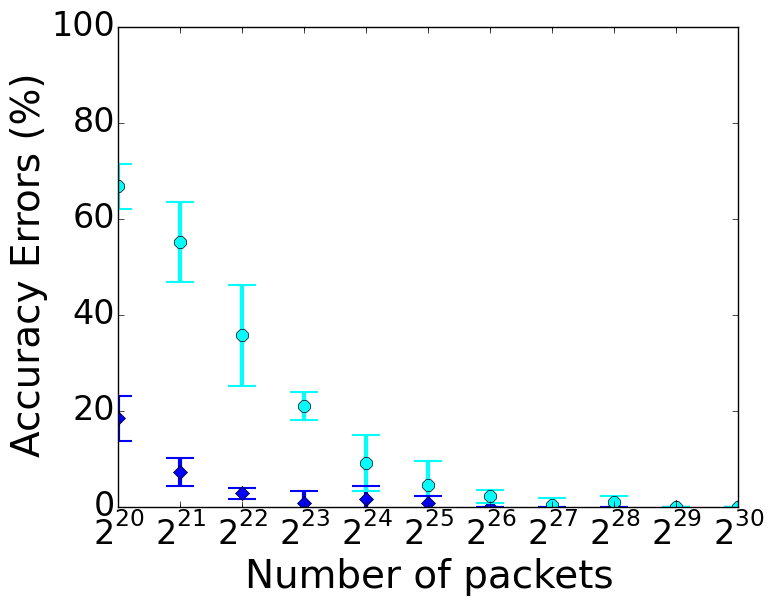}}  &
		\subfloat[Chicago16\label{ACC:CH16:2Dbyte} - 2D Bytes]{\includegraphics[width = \fourFigWidth]
			{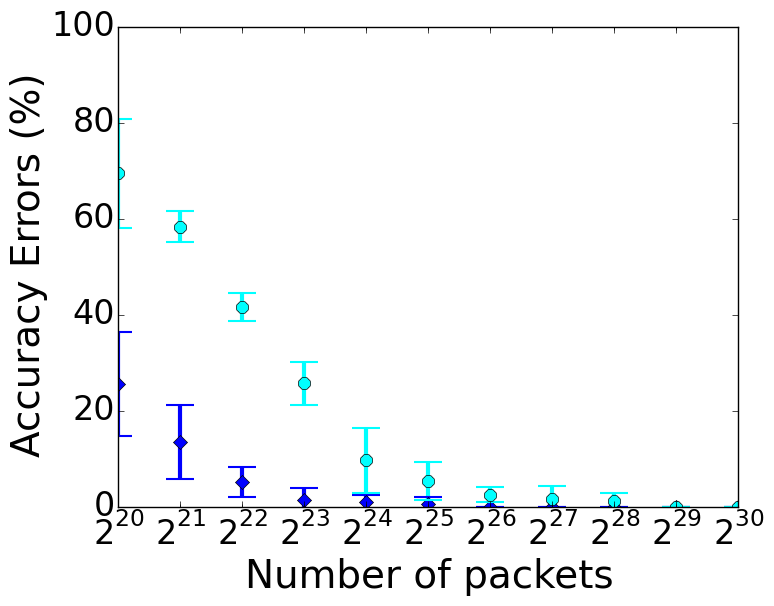}}  &
		\subfloat[SanJose13\label{ACC:SJ13:2Dbyte} - 2D Bytes]{\includegraphics[width = \fourFigWidth]
			{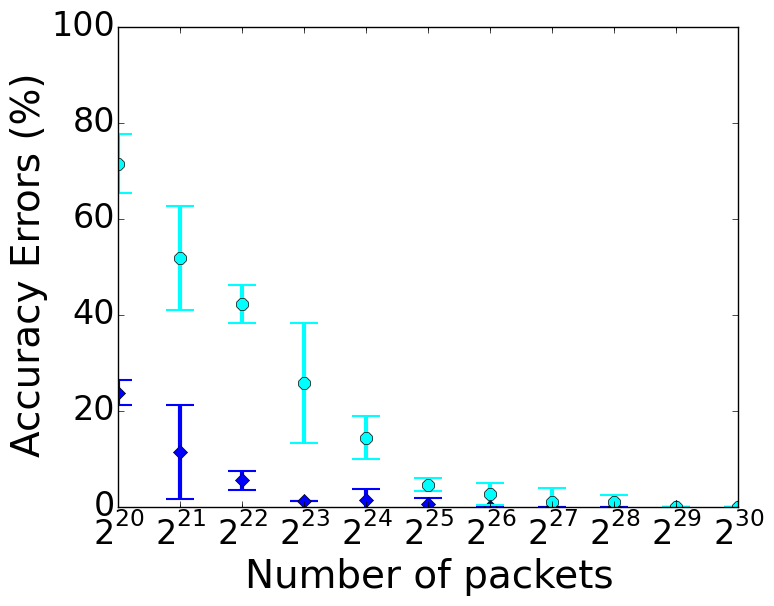}}  &
		\subfloat[SanJose14\label{ACC:SJ14:2Dbyte} - 2D Bytes]{\includegraphics[width = \fourFigWidth]
			{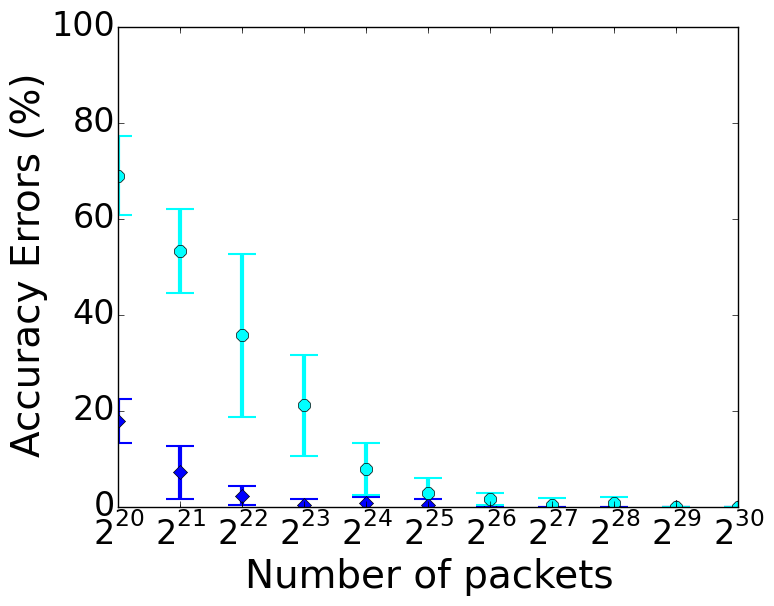}}  	
	\end{tabular}
	\caption{	\label{fig:acc} Accuracy error ratio -- HHH candidates whose frequency estimation error is larger than $N\epsilon$ ($\epsilon =0.001$).}		
	\captionsetup{justification=centering}
	\begin{tabular}{cccc}
		\multicolumn{4}{c}{\subfloat{\includegraphics[width = 7cm, height=0.6cm]
				{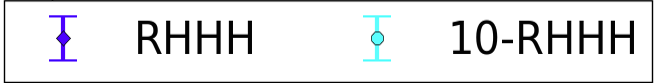}}}
		\tabularnewline
		\addtocounter{subfigure}{-1}
		\subfloat[Chicago15\label{COV:CH15:2Dbyte} - 2D Bytes]{\includegraphics[width = \fourFigWidth]
			{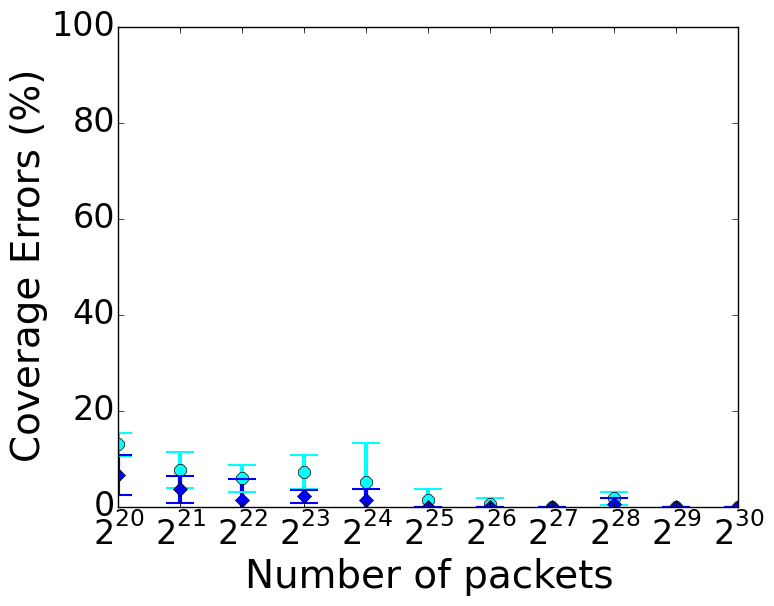}}  &
		\subfloat[Chicago16\label{COV:CH16:2Dbyte} - 2D Bytes]{\includegraphics[width = \fourFigWidth]
			{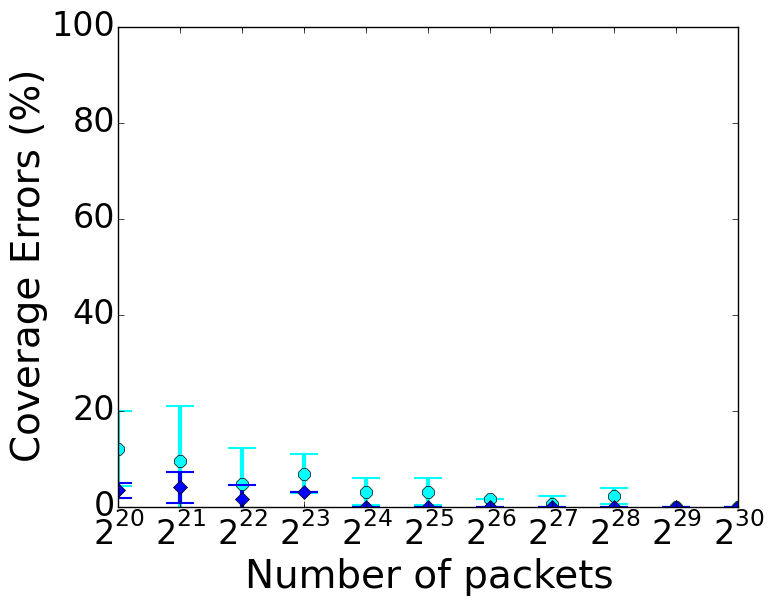}}  &
		\subfloat[SanJose13\label{COV:SJ13:2Dbyte} - 2D Bytes]{\includegraphics[width = \fourFigWidth]
			{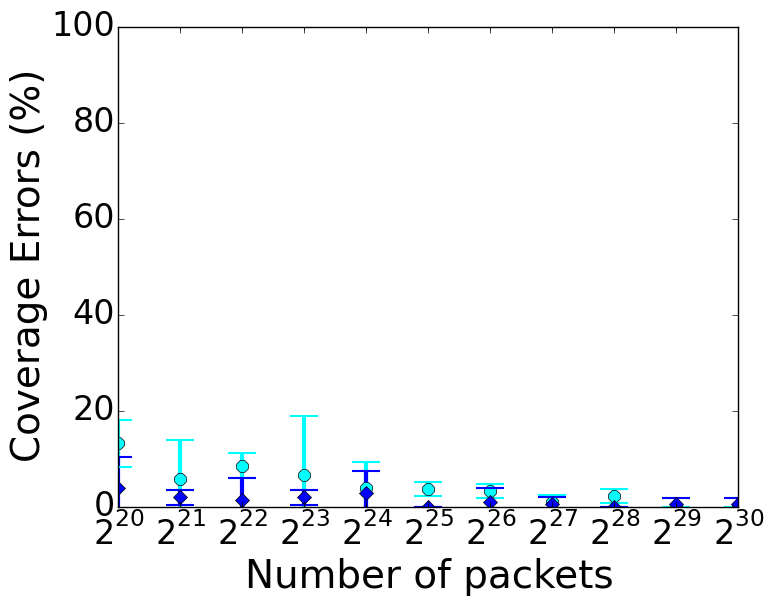}}  &
		\subfloat[SanJose14\label{COV:SJ14:2Dbyte} - 2D Bytes]{\includegraphics[width = \fourFigWidth]
			{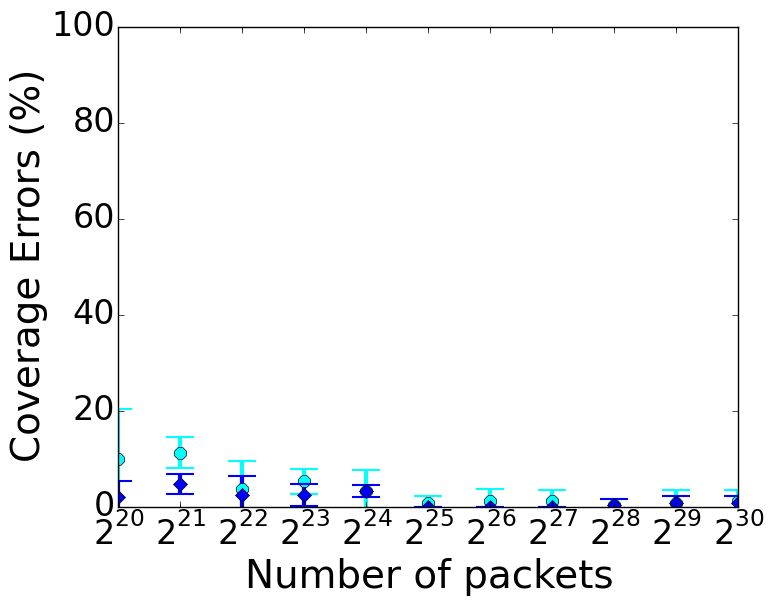}}   	
		
	\end{tabular}
	\caption{\label{fig:cov} The percentage of Coverage errors -- elements $q$ such that  $q\notin P$ and $C_{q\mid P}\ge N\theta$ (false negatives).
		}
\end{figure*}
\begin{figure*}[t]
	\begin{tabular}{ccc}
		\subfloat[\label{FPR:SJ:Byte}SanJose14 - 1D Bytes]{\includegraphics[width = \thirdFigWidth]	{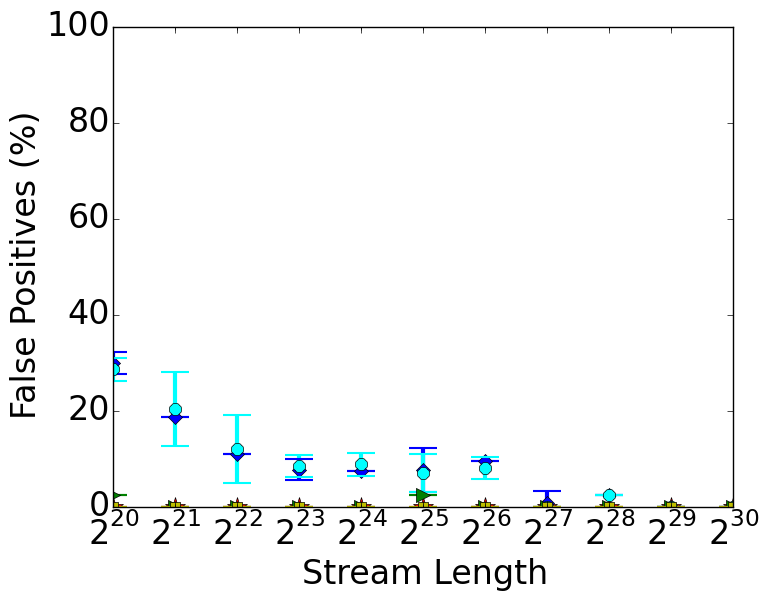}} &	
		\subfloat[SanJose14\label{FPR:SJ:bit} - 1D Bits]{\includegraphics[width = \thirdFigWidth]
			{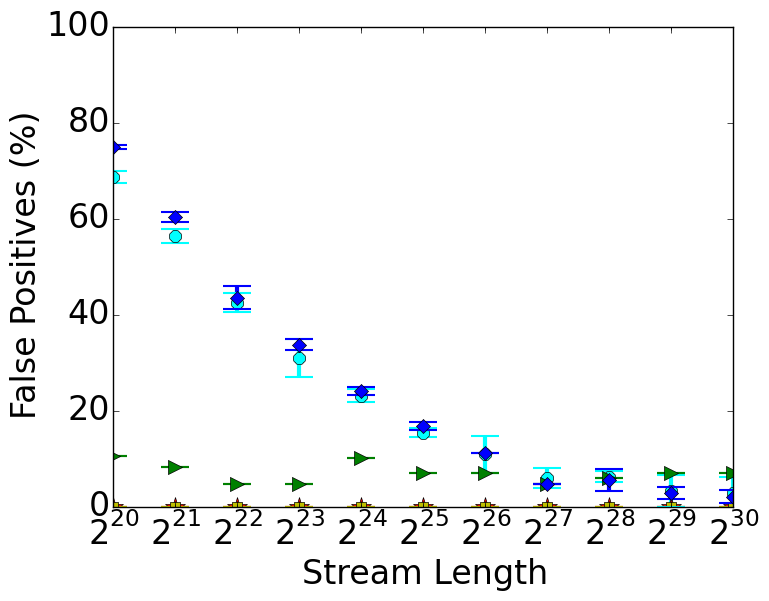}} &
		\subfloat[SanJose14\label{FPR:SJ:2Dbyte} - 2D Bytes]{\includegraphics[width = \thirdFigWidth]
			{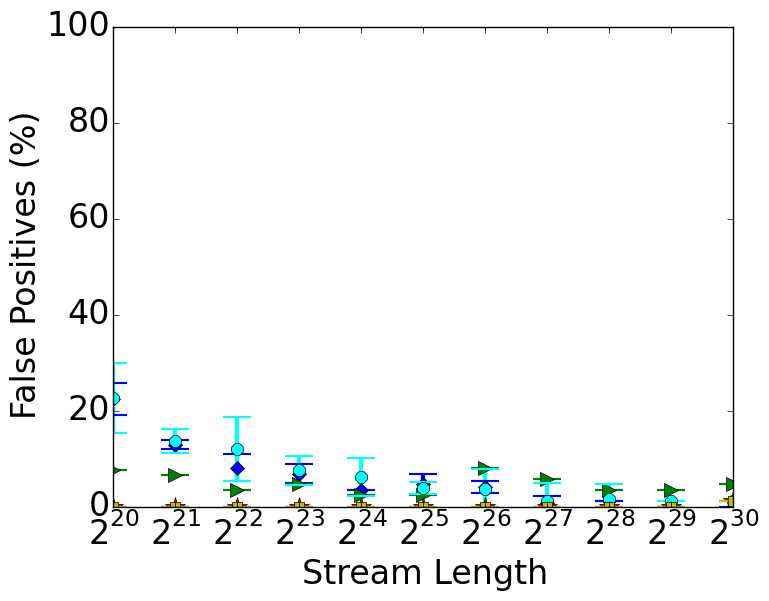}}
		\tabularnewline
		\multicolumn{3}{c}{\subfloat{\includegraphics[width = 17cm]
				{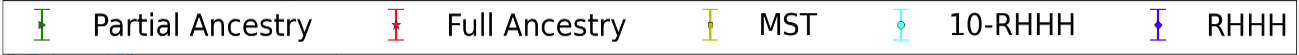}}}		
		\tabularnewline
		\addtocounter{subfigure}{-1}
		\subfloat[\label{FPR:CHI:byte}\label{FPR:`}Chicago16 - 1D Bytes]{\includegraphics[width = \thirdFigWidth]
			{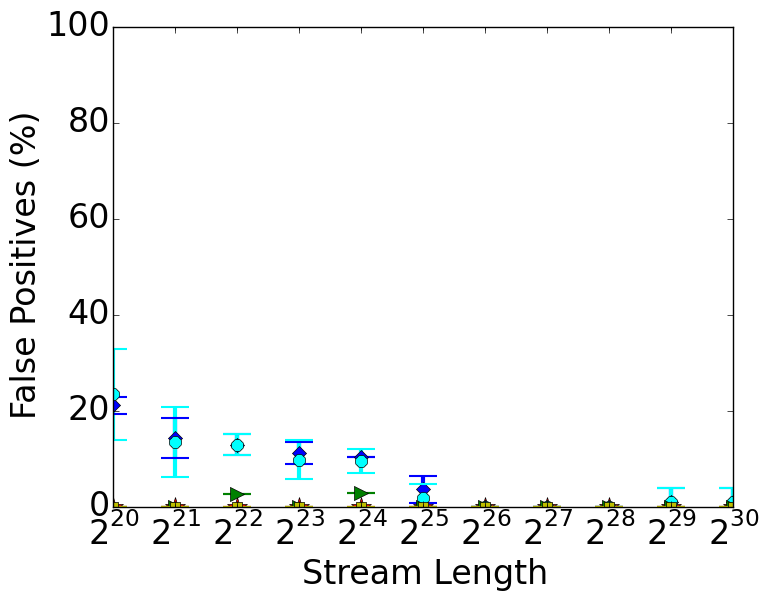}} &	
		\subfloat[\label{FPR:CHI:bit}Chicago16 - 1D Bits]{\includegraphics[width = \thirdFigWidth]
			{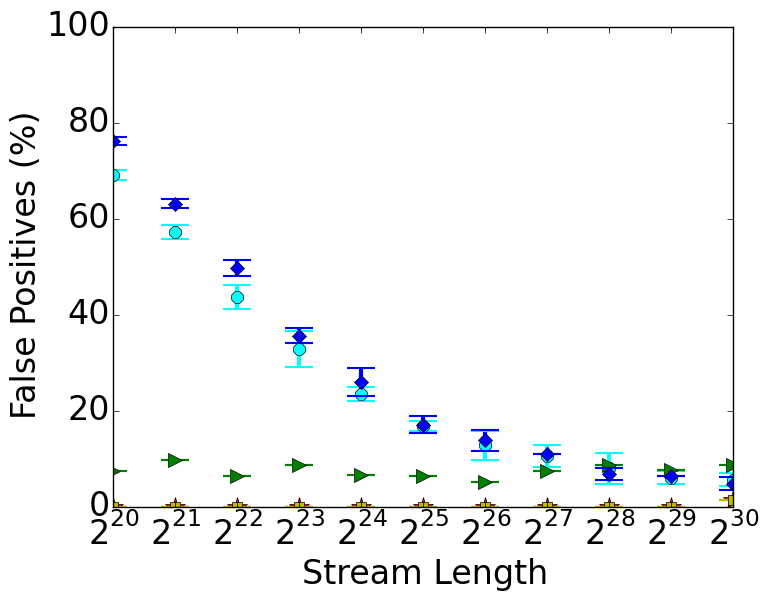}} &
		\subfloat[\label{FPR:CHI:2Dbyte}Chicago16 - 2D Bytes]{\includegraphics[width = \thirdFigWidth]
			{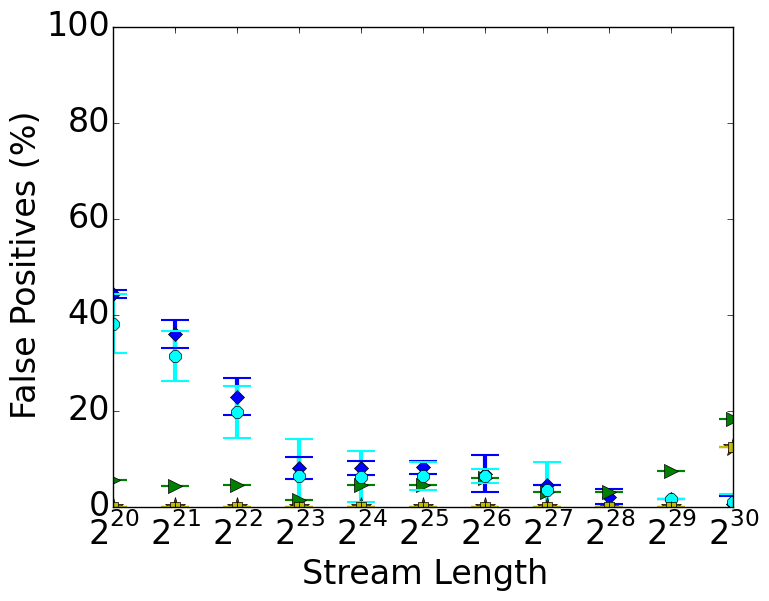}}
		
	\end{tabular}
	\caption{	\label{FPR:fig:FPR} False Positive Rate for different stream lengths.}
\end{figure*}

\begin{figure*}[t!]
	\begin{tabular}{ccc}
		\subfloat[\label{SJ:Byte}SanJose14 - 1D Bytes]{\includegraphics[width = \thirdFigWidth]
			{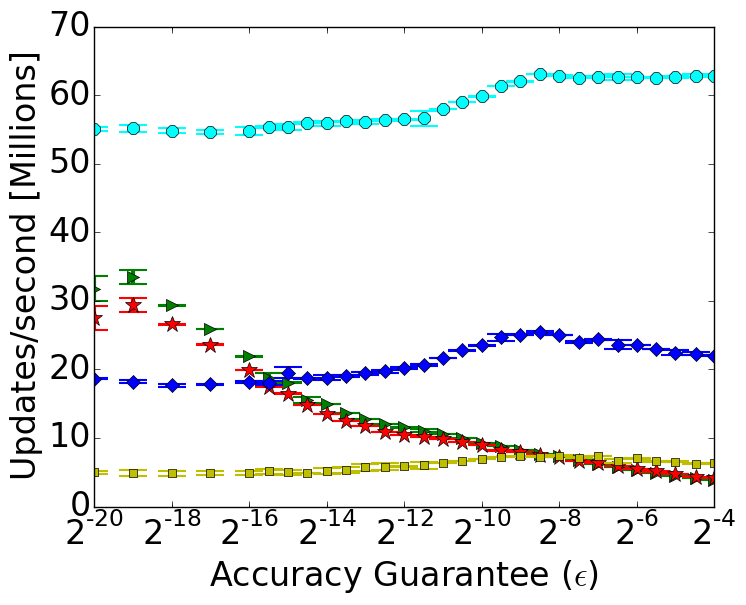}} &	
		\subfloat[SanJose14\label{SJ:bit} - 1D Bits]{\includegraphics[width = \thirdFigWidth]
			{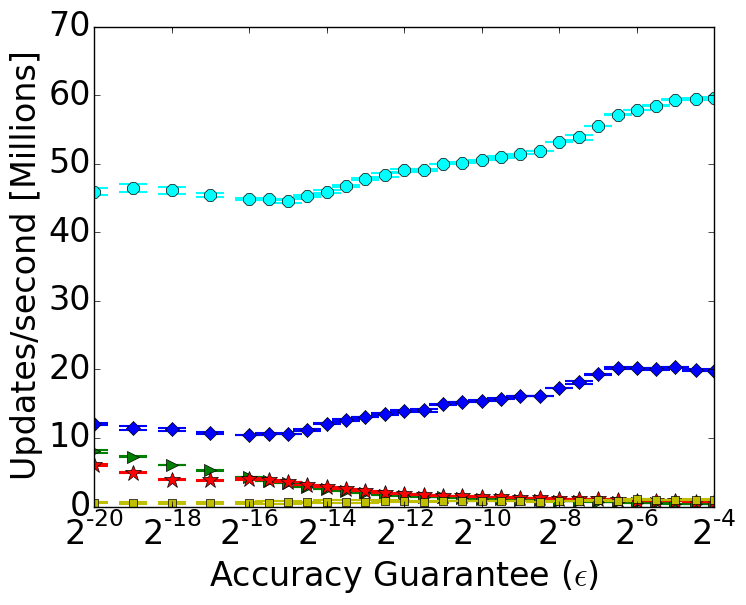}} &
		\subfloat[SanJose14\label{SJ:2Dbyte} - 2D Bytes]{\includegraphics[width = \thirdFigWidth]
			{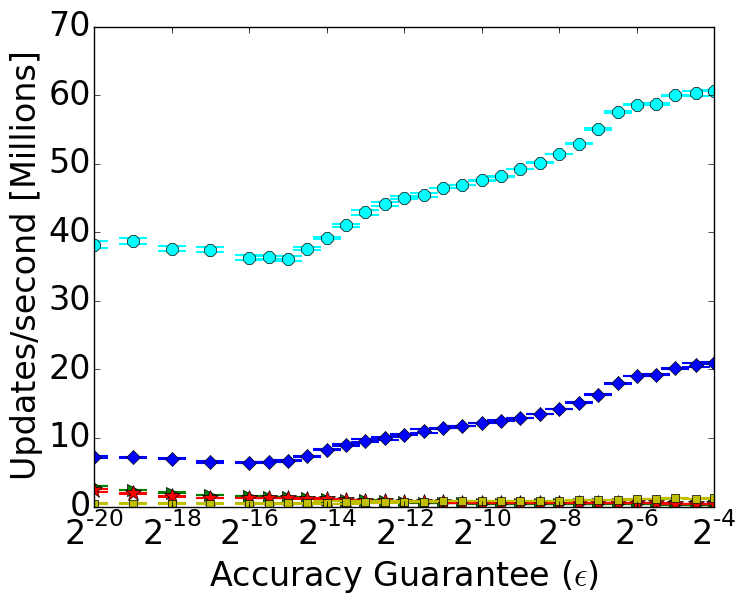}}
		\tabularnewline
		\multicolumn{3}{c}{\subfloat{\includegraphics[width = 17cm]
				{BarsLegend.png}}}		
		\tabularnewline
		\addtocounter{subfigure}{-1}
		\subfloat[\label{CHI:byte}\label{`}Chicago16 - 1D Bytes]{\includegraphics[width = \thirdFigWidth]
			{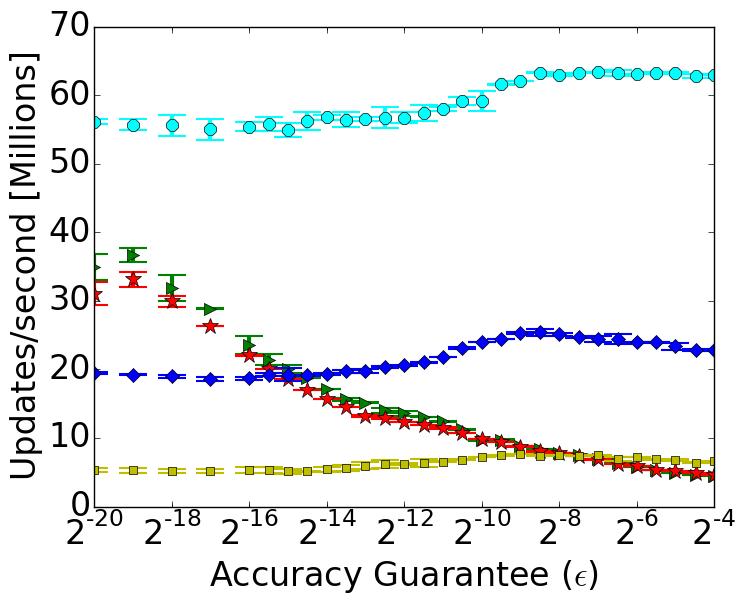}} &	
		\subfloat[\label{CHI:bit}Chicago16 - 1D Bits]{\includegraphics[width = \thirdFigWidth]
			{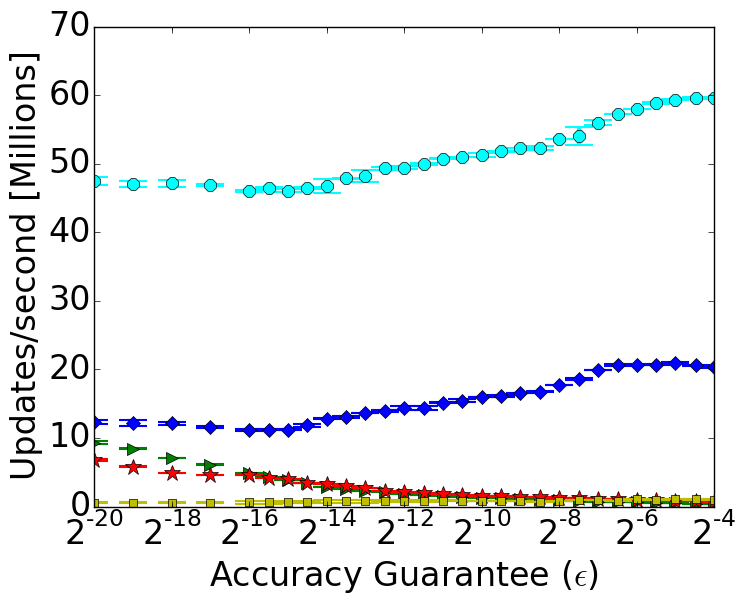}} &
		\subfloat[\label{CHI:2Dbyte}Chicago16 - 2D Bytes]{\includegraphics[width = \thirdFigWidth]
			{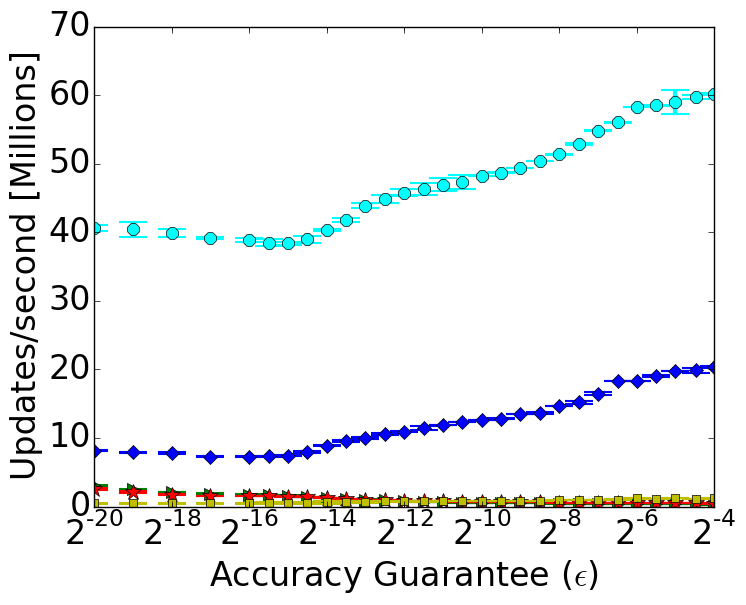}}
		
	\end{tabular}
	\caption{	\label{fig:speed} Update speed comparison for different hierarchical structures and workloads}
\end{figure*}


Our solution, \emph{Randomized HHH} (RHHH), updates \textbf{at most a single} randomly selected HH instance that operates in $O(1)$.
In contrast,~\cite{HHHMitzenmacher} updates \textbf{every} HH algorithm for each packet and thus operates in $O(H)$.

Specifically, for each packet, we randomize a number between 0 and $V$ and if it is smaller than $H$, we update the corresponding HH algorithm.
Otherwise, we ignore the packet.
Clearly, $V$ is a performance parameter: when $V=H$, every packet updates one of the HH algorithms whereas when $V\gg H$, most packets are ignored.
Intuitively, each HH algorithm receives a \emph{sample} of the stream.
We need to prove that given enough traffic, hierarchical heavy hitters can still be extracted.

Pseudocode of RHHH is given in Algorithm~\ref{alg:Skipper}.
RHHH uses the same algorithm for both one and two dimensions.
The differences between them are manifested in the $calcPred$ method.
Pseudocode of this method is found in Algorithm~\ref{alg:randHHH} for one dimension and in Algorithm~\ref{alg:randHHH2D} for two dimensions.

\begin{definition}
	The underlying estimation provides us with upper and lower estimates for the number of times prefix $p$ was updated ($X_p$).
    We denote: $\widehat{X^p}^{+}$ to be an upper bound for $X_p$ and $\widehat{X^p}^{-}$ to be a lower bound.
	For simplicity of notations, we define the following:\\
	$\widehat{f_p}\triangleq \widehat{X^p} V$ -- an estimator for $p$'s frequency.\\
	$\widehat{f_p^{+}}\triangleq \widehat{X^p}^{+} V$ -- an upper bound for $p$'s frequency.\\
	$\widehat{f_p^{-}}\triangleq \widehat{X^p}^{-} V$ -- a lower bound for $p$'s frequency.

\end{definition}
Note these bounds ignore the sample error that is accounted separately in the analysis.

The output method of RHHH starts with fully specified items and if their frequency is above $\theta N$, it adds them to $P$.
Then, RHHH iterates over their parent items and calculates a conservative estimation of their conditioned frequency with respect to $P$.
Conditioned frequency is calculated by an upper estimate to ($f_p^{+}$) amended by the output of the $calcPred$ method.
In a single dimension, we reduce the lower bounds of $p$'s closest predecessor HHHs.
In two dimensions, we use inclusion and exclusion principles to avoid double counting.
In addition, Algorithm~\ref{alg:randHHH2D} uses the notation of \emph{greater lower bound (glb)} that is formally defined in Definition~\ref{def:glb}.
Finally, we add a constant to the conditioned frequency to account for the sampling error.



\begin{definition}
	\label{def:glb}
	Denote $glb(h,h')$ the greatest lower bound of $h$ and $h'$.
	$glb(h,h')$ is a unique common descendant of $h$ and $h'$ s.t. $\forall p : \left(q\preceq p\right)\wedge \left(p \preceq h\right)\wedge \left(p \preceq h'\right) \Rightarrow p = q.$
	When $h$ and $h'$ have no common descendants, define $glb(h,h')$ as an item with count $0$.
\end{definition}

In two dimensions, $C_{p|P}$ is first set to be the upper bound on $p$'s frequency (Line~\ref{line:cp}, Algorithm~\ref{alg:Skipper}).
Then, we remove previously selected descendant heavy hitters (Line~\ref{alg:second}, Algorithm~\ref{alg:randHHH2D}).
Finally, we add back the common descendant (Line~\ref{alg:third}, Algorithm~\ref{alg:randHHH2D})).

Note that the work of~\cite{HHHMitzenmacher} showed that their structure extends to higher dimensions, with only a slight modification to the Output method to ensure that it conservatively estimates the conditioned count of each prefix.
As we use the same general structure, their extension applies in our case as well.

\newcommand{\FRHHH}{10-RHHH}
\section{Evaluation}
\label{sec:eval}

Our evaluation includes MST~\cite{HHHMitzenmacher}, the Partial and Full Ancestry~\cite{CormodeHHH} algorithms and two configurations of RHHH, one with $V=H$ (RHHH) and the other with $V=10\cdot H$ (10-RHHH).
RHHH performs a single update operation per packet while \FRHHH{} performs such an operation only for 10\% of the packets.
Thus, \FRHHH{} is considerably faster than RHHH but requires more traffic to converge.


The evaluation was performed on a single Dell 730 server running Ubuntu 16.04.01 release.
The server has 128GB of RAM and an Intel(R) Xeon(R) CPU E5-2667 v4 @ 3.20GHz processor.

Our evaluation includes four datasets, each containing a mix of 1 billion UDP/TCP and ICMP packets collected from major backbone routers in both Chicago~\cite{CAIDACH15,CAIDACH16} and San Jose~\cite{CAIDASJ13,CAIDASJ14} during the years 2014-2016.
We considered source hierarchies in byte (1D Bytes) and bit (1D Bits) granularities, as well as a source/destination byte hierarchy (2D Bytes).
Such hierarchies were also used by~\cite{HHHMitzenmacher,CormodeHHH}.
We ran each data point $5$ times and used two-sided Student's t-test to determine 95\% confidence intervals.

\subsection{Accuracy and Coverage Errors}
\label{sec:acc+cov-error}
RHHH has a small probability of both accuracy and coverage errors that are not present in previous algorithms.
Figure~\ref{fig:acc} quantifies the accuracy errors and Figure~\ref{fig:cov} quantifies the coverage errors.
As can be seen, RHHH becomes more accurate as the trace progresses.
Our theoretic bound ($\NB$ as derived in Section~\ref{sec:analysis} below) for these parameters is about 100 million packets for RHHH and about 1 billion packets for \FRHHH{}.
Indeed, these algorithms converge once they reach their theoretical bounds (see Theorem~\ref{thm:correctness}).



\subsection{False Positives}
Approximate HHH algorithms find all the HHH prefixes but they also return non HHH prefixes.
\emph{False positives} measure the ratio non HHH prefixes pose out of the returned HHH set.
Figure~\ref{FPR:fig:FPR} shows a comparative measurement of false positive ratios in the Chicago 16 and San Jose 14 traces. Every point was measured for $\epsilon=0.1\%$ and $\theta=1\%$.
As shown, for RHHH and \FRHHH{} the false positive ratio is reduced as the trace progresses.
Once the algorithms reach their theoretic grantees ($\NB$), the false positives are comparable to these of previous works.
In some cases, RHHH and \FRHHH{} even perform slightly better than the alternatives.



\subsection{Operation Speed}

Figure~\ref{fig:speed} shows a comparative evaluation of operation speed.
Figure~\ref{SJ:Byte}, Figure~\ref{SJ:bit} and Figure~\ref{SJ:2Dbyte} show the results of the San Jose 14 trace for 1D byte hierarchy ($H=5$), 1D bit hierarchy ($H=33$) and 2D byte hierarchy ($H=25$), respectively.
Similarly, Figure~\ref{CHI:byte}, Figure~\ref{CHI:bit} and Figure~\ref{CHI:2Dbyte} show results for the Chicago 16 trace on the same hierarchical domains.
Each point is computed for $250M$ long packet traces.
Clearly, the performance of RHHH and \FRHHH{} is relatively similar for a wide range of $\varepsilon$ values and for different data sets.
Existing works depend on $H$ and indeed run considerably slower for large $H$ values. 

Another interesting observation is that the Partial and Full Ancestry~\cite{CormodeHHH} algorithms improve when $\varepsilon$ is small.
This is because in that case there are few replacements in their trie based structure, as is directly evident by their $O(H\log(N\epsilon))$ update time, which is decreasing with $\epsilon$.
However, the effect is significantly lessened when $H$ is large.



RHHH and \FRHHH{} achieve speedup for a wide range of $\varepsilon$ values, while \FRHHH{} is the fastest algorithm overall.
For one dimensional byte level hierarchies, the achieved speedup is up to X3.5 for RHHH and up to X10 for \FRHHH{}.
For one dimensional bit level hierarchies, the achieved speedup is up to X21 for RHHH and up to X62 for \FRHHH{}.
Finally, for 2 dimensional byte hierarchies, the achieved speedup is up to X20 for RHHH and up to X60 for \FRHHH{}.
Evaluation on Chicago15 and SanJose13 yielded similar results, which are omitted due to lack of space.

\section{Virtual Switch Integration}
\label{sec:ovs}
This section describes how we extended Open vSwitch (OVS) to include approximate HHH monitoring capabilities. For completeness, we start with a short overview of OVS and then continue with our~evaluation.

\subsection{Open vSwitch Overview}
\label{apx:ovsOverview}

Virtual switching is a key building block in NFV environments, as it enables interconnecting multiple \emph{Virtual Network Functions} (VNFs) in service chains and enables the use of other routing technologies such as SDN. 
In practice, virtual switches rely on sophisticated optimizations to cope with the line rate.

Specifically, we target the DPDK version of OVS that enables the entire packet processing to be performed in user space. 
It mitigates overheads such as interrupts required to move from user space to kernel space. 
In addition, DPDK enables user space packet processing and provides direct access to NIC buffers without unnecessary memory copy. 
The DPDK library received significant engagement from the NFV industry~\cite{intelDpdk}.

%


%
The architectural design of OVS is composed of two main components: ovs-vswitchd and ovsdb-server. 
Due to space constraints, we only describe the vswitchd component. 
The interested reader is referred to \cite{ovs-2015-nsdi} for additional information.
The DPDK-version of the vswitchd module implements control and data planes in user space. 
Network packets ingress the datapath (dpif or dpif-netdev) either from a physical port connected to the physical NIC or from a virtual port connected to a remote host (e.g., a VNF). 
The datapath then parses the headers and determines the set of actions to be applied (e.g., forwarding or rewrite a specific~header).

\subsection{Open vSwitch Evaluation}

We examined two integration methods:
First, HHH measurement can be performed as part of the OVS dataplane.
That is, OVS updates each packet as part of its processing stage.
Second, HHH measurement can be performed in a separate virtual machine.
In that case, OVS forwards the relevant traffic to the virtual machine.
When RHHH operates with $V>H$, we only forward the sampled packets and thus reduce~overheads.

\subsubsection{OVS Environment Setup}
Our evaluation settings consist of two identical HP ProLiant servers with an Intel Xeon E3-1220v2 processor running at 3.1 Ghz with 8 GB RAM, an Intel 82599ES 10 Gbit/s network card and CentOS 7.2.1511 with Linux kernel 3.10.0 operating system.
The servers are directly connected through two physical interfaces.
We used Open vSwitch 2.5 with Intel DPDK 2.02, where NIC physical ports are attached using \emph{dpdk} ports.

One server is used as traffic generator while the other is used as \emph{Design Under Test (DUT)}.
Placed on the DUT, OVS receives packets on one network interface and then forwards them to the second one.
Traffic is generated using MoonGen traffic generator~\cite{MoonGen2015}, and we generate 1 billion UDP packets but preserve the source and destination IP as in the original dataset.
We also adjust the payload size to 64 bytes and reach 14.88 million packets per second (Mpps).


\begin{figure}[t]\centering
	\includegraphics[width = 0.8\columnwidth]{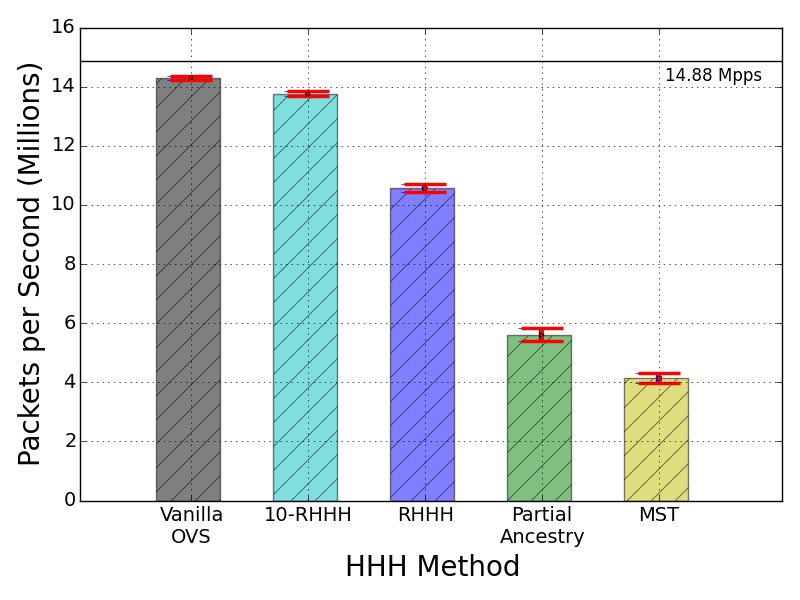}
	\caption{Throughput of dataplane implementations
		($\varepsilon = 0.001$, $\delta = 0.001$, 2D Bytes, Chicago 16).}
	\label{fig-bw1}
\end{figure}

\subsubsection{OVS Throughput Evaluation}
Figure \ref{fig-bw1} exhibits the throughput of OVS for dataplane implementations.
It includes our own \FRHHH{} (with V=10H) and RHHH (with V=H), as well as MST and Partial Ancestry.
Since we only have 10 Gbit/s links, the maximum achievable packet rate is 14.88~Mpps.

As can be seen, \FRHHH{} processes 13.8 Mpps, only 4\% lower than unmodified OVS.
RHHH achieves 10.6 Mpps, while the fastest competition is Partial Ancestry that delivers 5.6 Mpps.
Note that a 100 Gbit/s link delivering packets whose average size is 1KB only delivers $\approx$ 8.33 Mpps.
Thus, \FRHHH{} and RHHH can cope with the line speed.

Next, we evaluate the throughput for different $V$ values, from $V=H=25$ (RHHH) to $V=10\cdot H =250$ (\FRHHH{}).
Figure~\ref{fig:DPI} evaluates the dataplane implementation while Figure~\ref{fig:VMI} evaluates the distributed implementation.
In both figures, performance improves for larger $V$ value. 
In the distributed implementation, this speedup means that fewer packets are forwarded to the VM whereas in the dataplane implementation, it is linked to fewer processed~packets.
\begin{figure}[t]\centering
	{\includegraphics[width = 0.8\columnwidth]{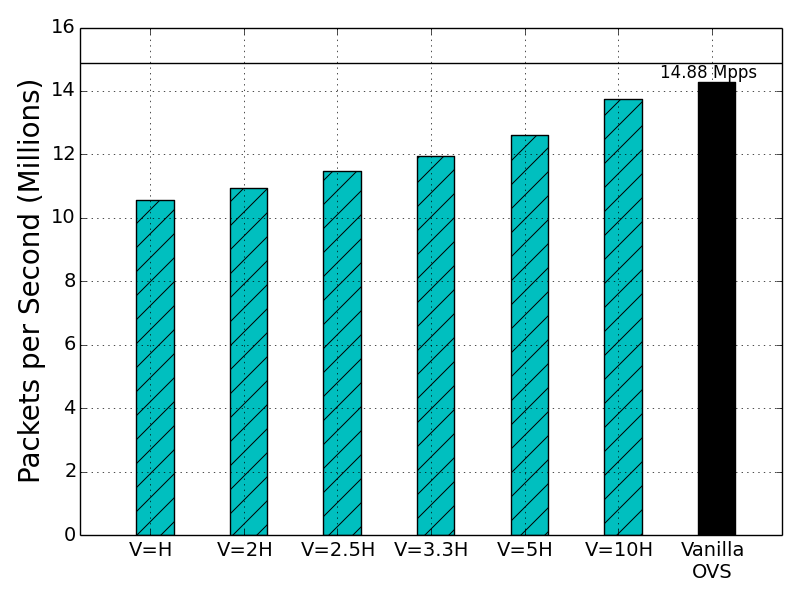}}
	\caption{Dataplane implementation\label{fig:DPI}}
\end{figure}
\begin{figure}[t]\centering
	{\includegraphics[width = 0.8\columnwidth]{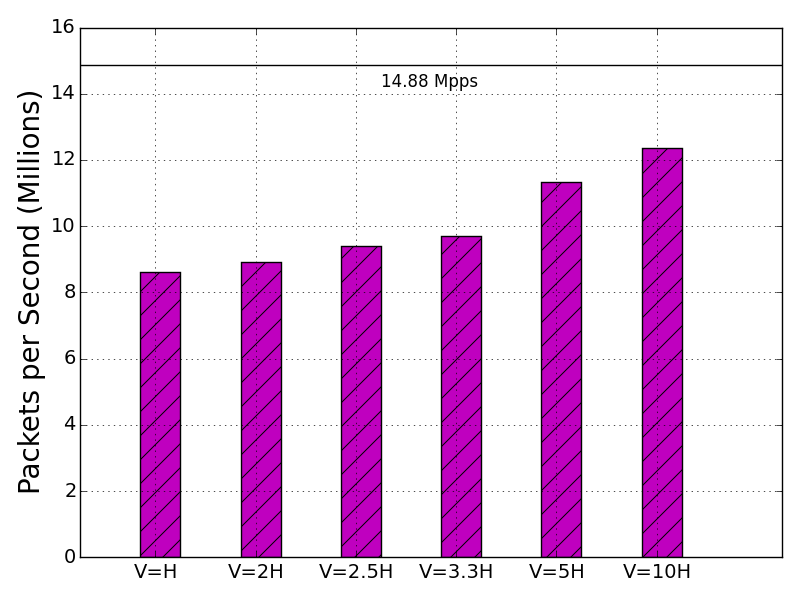}}
	\caption{Distributed implementation\label{fig:VMI}}
\end{figure}


Note that while the distributed implementation is somewhat slower, it enables the measurement machine to process traffic from multiple sources.

\section{Analysis}
\label{sec:analysis}
This section aims to prove that RHHH solves the {\sc$(\delta,\epsilon,\theta)-$approximate HHH} problem (Definition~\ref{def:deltaapproxHHH}) for one and two dimensional hierarchies.
Toward that end, Section~\ref{sec:analSamples} proves the accuracy requirement while Section~\ref{anal:randHHH} proves coverage.
Section~\ref{sec:RHHH-prop} proves that RHHH solves the {\sc$(\delta,\epsilon,\theta)-$approximate HHH} problem as well as its memory and update complexity.

We model the update procedure of RHHH as a balls and bins experiment where there are $V$ bins and $N$ balls.
Prior to each packet arrival, we place the ball in a bin that is selected uniformly at random.
The first $H$ bins contain an HH update action while the next $V-H$ bins are void.
When a ball is assigned to a bin, we either update the underlying HH algorithm with a prefix obtained from the packet's headers or ignore the packet if the bin is void.
Our first goal is to derive confidence intervals around the number of balls in a bin.

\begin{definition}
We define $X^{K}_i$ to be the random variable representing the number of balls from set $K$ in bin $i$, e.g., $K$ can be all packets that share a certain prefix, or a combination of multiple prefixes with a certain characteristic.
When the set $K$ contains all packets, we use the notation $X_i$.
\end{definition}

Random variables representing the number of balls in a bin are dependent on each other.
Therefore, we cannot apply common methods to create confidence intervals.
Formally, the dependence is manifested as:\\ $\sum\nolimits_1^{V} {{X_i}}  = N.$
This means that the number of balls in a certain bin is determined by the number of balls in all other bins.

Our approach is to approximate the balls and bins experiment with the corresponding Poisson one.
That is, analyze the Poisson case and derive confidence intervals and then use Lemma~\ref{lemma:rare} to derive a (weaker) result for the original balls and bins case.

We now formally define the corresponding Poisson model.
Let $Y_1^K,...,Y_{ V }^K$ s.t. $\{Y_i^K\} \sim Poisson\left( {\frac{K}{V}} \right)$ be \textbf{independent} Poisson random variables representing the number of balls in each bin from a set of balls $K$.
That is: $\{Y_i^K\} \sim Poisson\left( {\frac{K}{V}} \right).$

\begin{lemma}[Corollary 5.11, page 103 of~\cite{Mitzenmacher:2005:PCR:1076315}]
	\label{lemma:rare}
	Let $\mathfrak E$ be an event whose probability is either monotonically increasing or decreasing with the number of balls. If $\mathfrak E$ has probability $p$ in the Poisson case then $\mathfrak E$ has probability at most $2p$ in the exact case.
\end{lemma}


\subsection{Accuracy Analysis}
\label{sec:analSamples}

We now tackle the accuracy requirement from Definition~\ref{def:deltaapproxHHH}.
That is, for every HHH prefix ($p$), we need to~prove:  $$\Pr \left( {\left| {{f_p} - \widehat {{f_p}}} \right| \le \varepsilon N} \right) \ge 1 - \delta.$$

In RHHH, there are two distinct origins of error.
Some of the error comes from fluctuations in the number of balls per bin while the approximate HH algorithm is another source of error.

We start by quantifying the balls and bins error.
Let $Y^{p}_i$ be the Poisson variable corresponding to prefix $p$.
That is, the set $p$ contains all packets that are generalized by prefix $p$.
Recall that $f_p$ is the number of packets generalized by $p$ and therefore: $E(Y^{p}_i) = \frac{f_p}{V}.$

We need to show that with probability $1-\delta_s$,  $Y^{p}_i$ is within $\epsilon_s N$ from $E(Y^{p}_i)$.
Fortunately, confidence intervals for Poisson variables are a well studied~\cite{19WaysToPoisson} and we use the method of~\cite{Wmethod} that is quoted in
Lemma~\ref{lemma:poissonConfidence}.
\begin{lemma}
	\label{lemma:poissonConfidence}
	Let $X$ be a Poisson random variable, then
	\[\Pr \left( {\left| {X - E\left( X \right)} \right| \ge {Z_{1-\delta }}\sqrt {E\left( X \right)} } \right) \le \delta,\]
where $Z_\alpha$ is the $z$ value that satisfies $\phi(z)=\alpha$ and $\phi(z)$ is the density function of the normal distribution with mean $0$ and standard deviation of $1$.
\end{lemma} Lemma~\ref{lemma:poissonConfidence}, provides us with a confidence interval for Poisson variables, and enables us to tackle the main accuracy result.
\begin{theorem}
	\label{thm:pusmain}
	If $N \ge {Z_{1 - \frac{{{\delta _s}}}{2}}}V{\varepsilon_s}^{ - 2}$ then
	\[\Pr \left( {\left| {{X_i}^pH - {f_p}} \right| \ge {\varepsilon _s}N} \right) \le {\delta _s}.\]
\end{theorem}

\begin{proof}
We use Lemma~\ref{lemma:poissonConfidence} for $\frac{\delta_s}{2}$ and get:
\[\Pr \left( \left| {{Y_i}^p - \frac{{{f_p}}}{V}} \right| \ge {Z_{1 - \frac{\delta _s}{2}}\sqrt {\frac{{{f_p}}}{V}} } \right) \le \frac{\delta_s}{2} .\]
To make this useful, we trivially bind $f_p \le N$ and get
\[\Pr \left( \left| {{Y_i}^p - \frac{{{f_p}}}{V}} \right| \ge {Z_{1 - \frac{\delta _s}{2}}\sqrt {\frac{{{N}}}{V}} } \right) \le \frac{\delta_s}{2} .\]

However, we require error of the form $\frac{\epsilon_s \cdot N}{V}$.
\[\begin{array}{l}
{\varepsilon _s}N{V^{ - 1}} \ge {Z_{1 - \frac{{{\delta _s}}}{2}}}{V^{ - 0.5}}{N^{0.5}}\\
{N^{0.5}} \ge {Z_{1 - \frac{{{\delta _s}}}{2}}}{V^{0.5}}{\varepsilon _s}^{ - 1}\\
N \ge {Z_{1 - \frac{{{\delta _s}}}{2}}}V{\varepsilon_s}^{ - 2} .
\end{array}\]
Therefore, when $N \ge {Z_{1 - \frac{{{\delta _s}}}{2}}}V{\varepsilon_s}^{ - 2}$,  we have that:
\[\Pr \left( {\left| {{Y_i}^p - \frac{{{f_p}}}{V}} \right| \ge \frac{{{\varepsilon _s}N}}{V}} \right) \le \frac{{{\delta _s}}}{2} .\]
We multiply by $V$ and get:
$$\Pr \left( {\left| {{Y_i}^pV - {f_p}} \right| \ge {\varepsilon _s}N} \right) \le \frac{{{\delta _s}}}{2} .$$
Finally, since $Y_i^{p}$ is monotonically increasing with the number of balls ($f_p$), we apply Lemma~\ref{lemma:rare} to conclude that\\
$$\Pr \left( {\left| {{X_i}^pV - {f_p}} \right| \ge {\varepsilon _s}N} \right) \le {\delta _s}.$$
\end{proof}
To reduce clutter, we denote $\NB \triangleq \NBound$.
Theorem~\ref{thm:pusmain} proves that the desired sample accuracy is achieved once $N>\NB$.

It is sometimes useful to know what happens when $N<\NB$.
For this case, we have Corollary~\ref{cor:epsN}, which is easily derived from Theorem~\ref{thm:pusmain}.
We use the notation $\varepsilon_s(N)$ to define the actual sampling error after $N$ packets.
Thus, it assures us that when $N<\NB$, $\varepsilon_s(N)>\varepsilon_s$.
It also shows that $\varepsilon_s(N)<\varepsilon_s$ when $N>\NB$.
Another application of Corollary~\ref{cor:epsN} is that given a measurement interval $N$, we can derive a value for $\varepsilon_s$ that assures correctness.
For simplicity, we continue with the notion of $\varepsilon_s$.

\begin{corollary}
	\label{cor:epsN}
	 ${\varepsilon _s}\left( N \right) \ge \sqrt {\frac{{{Z_{1 - \frac{{{\delta _s}}}{2}}}V}}{N}} .$
\end{corollary}

The error of approximate HH algorithms is proportional to the number of updates.
Therefore, our next step is to provide a bound on the number of updates of an arbitrary HH algorithm.
Given such a bound, we configure the algorithm to compensate so that the accumulated error remains within the guarantee even if the number of updates is larger than average.

\begin{corollary}
	\label{cor:oversample}
	Consider the number of updates for a certain lattice node ($X_i$).
	If $N>\NB$, then \[\Pr \left( {{X_i} \le \frac{N}{V}\left( {1 + {\varepsilon _s}} \right)} \right) \ge 1 - {\delta _s}.\]
\end{corollary}

\begin{proof}
	 We use Theorem~\ref{thm:pusmain} and get: \\
$\Pr \left( {\left| {{X_i} - \frac{N}{V}} \right| \ge {\varepsilon _s}N} \right) \le {\delta _s}.$
This implies that:\\
$\Pr \left( {{X_i} \le \frac{N}{V}\left( {1 + {\varepsilon _s}} \right)} \right) \ge 1 - {\delta _s},$
completing the proof.
\end{proof}

We explain now how to configure our algorithm to defend against situations in which a given approximate HH algorithm might get too many updates, a phenomenon we call \emph{over sample}.
Corollary~\ref{cor:oversample} bounds the probability for such an occurrence, and hence we can slightly increase the accuracy so that in the case of an over sample, we are still within the desired limit.
We use an algorithm ($\mathbb A$) that solves the {\sc {$(\varepsilon_a, \delta_a)$ - Frequency Estimation}} problem.
We define $\varepsilon_a' \triangleq \frac{\varepsilon_a}{1+\varepsilon_s}$.
According to Corollary~\ref{cor:oversample}, with probability $1-\delta_s$, the number of sampled packets is at most $(1+\varepsilon_s)\frac{N}{V}.$
By using the union bound and with probability $1-\delta_a-\delta_s$ we get:
\[\left| {{X^p} - \widehat {{X^p}}} \right| \le {\varepsilon _{a'}}\left( {1 + {\varepsilon _s}} \right)\frac{N}{V} = \frac{{{\varepsilon _a}\left( {1 + {\varepsilon _s}} \right)}}{{1 + {\varepsilon _s}}}\frac{N}{V} = {\varepsilon _a}\frac{N}{V}.\]
For example, Space Saving requires $1,000$ counters for $\epsilon_a=0.001$.
If we set $\epsilon_s = 0.001$, we now require $1001$ counters.
Hereafter, we assume that the algorithm is configured to accommodate these over samples.

\begin{theorem}
	\label{thm:PUSCombined}
	Consider an algorithm ($\mathbb{A}$) that solves the {\sc {$(\epsilon_a, \delta_a)$ - Frequency Estimation}} problem.
	If $N > \NB$, then for  $\delta \ge \delta_a + 2 \cdot \delta_s$ and $\epsilon \ge \epsilon_a + \epsilon_s$, $\mathbb{A}$ solves {\sc {$(\epsilon, \delta)$ - Frequency Estimation}}.
\end{theorem}

\begin{proof}
	As $N > \NB$, we use Theorem~\ref{thm:pusmain}.
	That is, the input solves {\sc {$(\epsilon, \delta)$ - Frequency Estimation}}.
	\begin{equation}
	\label{eq:delta2}
	\Pr \left[ {\left| {{f_p} - {X_p}V} \right| \ge {\varepsilon _s}N} \right] \le {\delta _s}.
	\end{equation}
	
	$\mathbb{A}$ solves the {\sc {$(\epsilon_a, \delta_a)$ - Frequency Estimation}} problem and provides us with an estimator $\widehat{X^p}$ that approximates $X^p$ --  the number of updates for prefix $p$.
	According to Corollary~\ref{cor:oversample}:
	\[\Pr \left( {\left| {{X^p} - \widehat {{X^p}}} \right| \le \frac{{{\varepsilon _a}N}}{V}} \right) \ge 1 - {\delta _a} - {\delta _s},\]
	and multiplying both sides by $V$ gives us:
	\begin{equation}
	\label{eq:nodelta2}
\Pr \left( {\left| {{X^p}V - \widehat {{X^p}}V} \right| \ge {\varepsilon _a}N} \right) \le {\delta _a} + {\delta _s}.
	\end{equation}
	We need to prove that: $\Pr \left( {\left| {{f_p} - \widehat {{X^p}}V} \right| \le \varepsilon N} \right) \ge 1 - \delta$.
	Recall that: $f_p = E(X^p)V$ and that $\widehat{f_p} = \widehat{X^p}V$ is the estimated frequency of $p$.
	Thus,
	\small
	\begin{align}
	&\Pr \left( {\left| {{f_p} - \widehat{f_p}} \right| \ge \varepsilon N} \right) = \Pr \left( {\left| {{f_p} - \widehat {{X^p}}V} \right| \ge \varepsilon N} \right)\notag\\
	=& \Pr \left( {\left| {{f_p} + \left( {{X^p}{V} - {X^p}{V}} \right) - {V}\widehat {{X^p}}} \right| \ge (\epsilon_a+\epsilon_s) N} \right)\label{eq:separation}
	\\ \le&\Pr \left( \left[{\left| {{f_p} - {X^p}{V}} \right| \ge {\varepsilon _s}N} \right]\vee  \left[{\left| {{X^p}{V} - \widehat {{X^p}}}{V} \right| \ge {\varepsilon _a}N}\right] \right)\notag,
	\end{align}\normalsize
	where the last inequality follows from the fact that in order for the error of~\eqref{eq:separation} to exceed $\epsilon N$, at least one of the events has to occur.
	We bound this expression using the Union bound.
\[\begin{array}{l}
\Pr \left( {\left| {{f_p} - \widehat {{f_p}}} \right| \ge \varepsilon N} \right) \le \\
\Pr \left( {\left| {{f_p} - {X^p}V} \right| \ge {\varepsilon _s}N} \right) + \Pr \left( {\left| {{X^p}V - \widehat {{X^p}}H} \right| \ge {\varepsilon _a}N} \right)  \\
\le{\delta _a} + 2{\delta _s},
\end{array}\]
where the last inequality is due to equations~\ref{eq:delta2} and~\ref{eq:nodelta2}.
\end{proof}

An immediate observation is that Theorem~\ref{thm:PUSCombined} implies accuracy, as it guarantees that with probability $1-\delta$ the estimated frequency of any prefix is within $\varepsilon N$ of the real frequency while the accuracy requirement only requires it for prefixes that are selected as HHH.

\begin{lemma}
	\label{lemma:accuracy}
	If $N > \NB$, then Algorithm~\ref{alg:Skipper} satisfies the accuracy constraint for $\delta = \delta_a+2\delta_s$ and $\epsilon = \epsilon_a+\epsilon_s$.
\end{lemma}

\begin{proof}
	The proof follows from Theorem~\ref{thm:PUSCombined}, as the frequency estimation of a prefix depends on a single HH~algorithm.
\end{proof}

\subsubsection*{Multiple Updates}
One might consider how RHHH behaves if instead of updating at most $1$ HH instance, we update $r$ independent instances. This implies that we may update the same instance more than once per packet. 
Such an extension is easy to do and still provides the required guarantees. Intuitively, this variant of the algorithm is what one would get if each packet is duplicated $r$ times. The following corollary shows that this makes RHHH converge $r$ times faster. 
\begin{corollary}
	Consider an algorithm similar to $RHHH$ with $V=H$, but for each packet we perform $r$ independent update operations. 
	If $N > \frac{\NB}{r}$, then this algorithm satisfies the accuracy constraint for $\delta = \delta_a+2\delta_s$ and $\epsilon = \epsilon_a+\epsilon_s$.
\end{corollary}
\begin{proof}
	Observe that the new algorithm is identical to running RHHH on a stream ($\mathcal{S'}$) where each packet in $\mathcal{S}$ is replaced by $r$ consecutive packets.
    Thus, Lemma~\ref{lemma:accuracy} guarantees that accuracy is achieved for $\mathcal{S'}$ after $\NB$ packets are processed.
    That is, it is achieved for the original stream ($\mathcal{S}$) after $N >\frac{\NB}{r}$ packets.
\end{proof}

\subsection{Coverage Analysis}
\label{anal:randHHH}

Our goal is to prove the coverage property of Definition~\ref{def:deltaapproxHHH}.
 That is:
$\Pr \left( \widehat {C_{q|P}} \ge C_{q|P} \right) \ge 1-\delta.$
Conditioned frequencies are calculated in a different manner for one and two dimensions.
Thus, Section~\ref{subsec:one} deals with one dimension and Section~\ref{subsec:two} with two.

We now present a common definition of the best generalized prefixes in a set.
\begin{definition}[Best generalization]
	\label{def:bestG}
	Define $G(q|P)$  as the set $\left\{ {p:p \in P,p \prec q,\neg \exists p' \in P:q \prec p' \prec p} \right\}$.
	Intuitively, $G(q|P)$ is the set of prefixes that are best generalized by $q$.
	That is, $q$ does not generalize any prefix that generalizes one of the prefixes in $G(q|P)$.
\end{definition}

\subsubsection{One Dimension}
\label{subsec:one}
We use the following lemma for bounding the error of our conditioned count estimates.
\begin{lemma}
	\label{lemma:cp}(\cite{HHHMitzenmacher})
	In one dimension, $${C_{q\mid P}} = {f_q} - \sum\nolimits_{h \in G(q|P)} {{f_h}} .$$
	\normalsize
\end{lemma}

Using Lemma~\ref{lemma:cp}, it is easier to establish that the conditioned frequency estimates calculated by Algorithm~\ref{alg:Skipper} are conservative.

\begin{lemma}
\label{lemma:sq}
The conditioned frequency estimation of Algorithm~\ref{alg:Skipper} is:
\[\widehat{C_{q|P}} = \widehat{f_q}^{+}-\sum\nolimits_{h \in G\left( {q|P} \right)} {\widehat{f_h}^- } +  2{Z_{1 - \delta }}\sqrt {N V} .\]
\end{lemma}
\begin{proof}
Looking at Line~\ref{line:cp} in Algorithm~\ref{alg:Skipper}, we get that: $$\widehat{C_{q|P}} = \widehat{f_q}^{+} + calcPred(q,P).$$
That is, we need to verify that the return value $calcPred(q,P)$ in one dimension (Algorithm~\ref{alg:randHHH}) is $\sum\nolimits_{h \in G\left( {q|P} \right)} {\widehat{f_h}^- }$.
This follows naturally from that algorithm.
Finally, the addition of $2{Z_{1 - \delta }}\sqrt {NV}$ is due to line~\ref{line:accSample}.
\end{proof}

In deterministic settings, $\widehat{f_q}^{+} -\sum\nolimits_{h \in G\left( {q|P} \right)} {\widehat{f_h}^- }$ is a conservative estimate since ${\widehat {{f_q}}^ + } \ge {f_q}$ and $f_h < \widehat{f_h}^-$.
In our case, these are only true with regard to the sampled sub-stream and the addition of $2{Z_{1 - \delta }}\sqrt {NV}$ is intended to compensate for the randomized process.


Our goal is to show that $\Pr \left(\widehat {C_{q|P}} > C_{q|P}\right) \ge 1-\delta$.
That is, the conditioned frequency estimation of Algorithm~\ref{alg:Skipper} is probabilistically conservative.



\begin{theorem}
	\label{thm:underCP}
$\Pr \left( \widehat {C_{q|P}} \ge C_{q|P} \right) \ge 1-\delta.$
\end{theorem}

\begin{proof}
	Recall that:  $$\widehat {{C_{q|P}}} = \widehat f_q^ +  - \sum\limits_{h \in G\left( {q|P} \right)} {\widehat f_h^ - + 2{Z_{1-\frac{\delta }{8}}}\sqrt {NV} }.$$

	We denote by $K$ the set of packets that may affect $\widehat {{C_{q|P}}}$.
	We split $K$ into two sets: $K^{+}$ contains the packets that may positively impact $\widehat {{C_{q|P}}}$ and $K^{-}$ contains the packets that may negatively impact it.
	
	 We use $K^{+}$ to estimate the sample error in $\widehat{ f_q}$ and $K^{-}$ to estimate the sample error in $\sum\limits_{h \in G\left( {q|P} \right)} {\widehat f_h^ -}$.
	 The positive part is easy to estimate.
	 In the negative, we do not know exactly how many bins affect the sum.
	 However, we know for sure that there are at most $N$.
	 We define the random variable $Y^K_+$ that indicates the number of balls included in the positive sum.
	  We  invoke Lemma~\ref{lemma:poissonConfidence} on $Y^{K}_+$.
	  For the negative part, the conditioned frequency is positive so $E\left(Y^K_-\right)$ is at most $\frac{N}{V}$. Hence,
	$\Pr \left( \left| {Y_K^ +  - E\left( {Y_K^ + } \right)} \right| \ge {Z_{1-\frac{\delta }{8}}}\sqrt \frac{N }{V} \right) \le \frac{\delta }{4}.$
	Similarly, we use Lemma~\ref{lemma:poissonConfidence} to bound the error of $Y_K^ -$:
	$$\Pr \left( {\left| {Y_K^ -  - E\left( {{Y_K}^ - } \right)} \right| \ge {Z_{1-\frac{\delta }{8}}}\sqrt  \frac{N}{V} } \right) \le \frac{\delta }{4}.$$\\
	$Y^{K}_+$ is monotonically increasing with any ball and $Y_{K}^-$ is monotonically decreasing with any ball.
	Therefore,  we can apply Lemma~\ref{lemma:rare} on each of them and conclude:
\[\begin{array}{l}
\Pr \left( {\widehat {{C_{q|P}}} \ge {C_{q|P}}} \right)\le\\
 2\Pr \left( {H\left( {Y_K^ -  + Y_K^ + } \right)   \ge VE\left( {Y_K^ -  + Y_K^ + } \right) + 2{Z_{1 - \frac{\delta }{8}}}\sqrt {NV} } \right)\\
\le 1-2\frac{\delta }{2} = 1-\delta.
\end{array}\]
\end{proof}
\normalsize
\begin{theorem}
	If $N > \NB$, Algorithm~\ref{alg:Skipper} solves the {\sc$(\delta, \varepsilon, \theta)$ - Approximate HHH} problem for $\delta = \delta_a + 2\delta_s$ and $\varepsilon = \varepsilon_s + \varepsilon_a$.
\end{theorem}
\begin{proof}
We need to show that the accuracy and coverage guarantees hold.
Accuracy follows from Lemma~\ref{lemma:accuracy} and coverage follows from Theorem~\ref{thm:underCP} that implies that for every non heavy hitter prefix (q),  $\widehat{C_{q|P}}<\theta N$ and thus:  $$\Pr \left( {{C_{q|P}} < \theta N} \right) \ge 1 - \delta.$$
\end{proof}


\subsubsection{Two Dimensions}
\label{subsec:two}
Conditioned frequency is calculated differently for two dimensions, as
we use inclusion/exclusion principles and we need to show that these calculations are sound too.
We start by stating the following lemma:
\begin{lemma}
	\label{lemma:cp2d}(\cite{HHHMitzenmacher})
	In two dimensions,
	
	\[{C_{q|P}} = {f_q} - \sum\limits_{h \in G\left( {q|P} \right)} {{f_h}}  + \sum\limits_{h,h' \in G\left( {q|P} \right)} {f_{{\rm{glb}}\left( {h,h'} \right)}^{}} .\]
	\normalsize
\end{lemma}
\noindent In contrast, Algorithm~\ref{alg:Skipper} estimates the conditioned frequency as:
\begin{lemma}
	\label{lemma:algCF2D}
	In two dimensions, Algorithm~\ref{alg:Skipper} calculates conditioned frequency in the following manner: \small\[\widehat {{C_{q|P}}} = \hat f_q^ +  - \sum\limits_{h \in G\left( {q|P} \right)} {\hat f_h^ - }  + \sum\limits_{h,h' \in G\left( {q|P} \right)} {\hat f_{{\rm{glb}}\left( {h,h'} \right)}^ + }  + 2{Z_{1 - \frac{\delta }{8}}}\sqrt {NV} .\]
	\normalsize
\end{lemma}
\begin{proof}
The proof follows from Algorithm~\ref{alg:Skipper}.
Line~\ref{line:cp} is responsible for the first element $\widehat{f_q^{+}}$ while Line~\ref{line:accSample} is responsible for the last element.
The rest is due to the function calcPredecessors in Algorithm~\ref{alg:randHHH2D}.
\end{proof}

\begin{theorem}
	\label{thm:conservativeCP}
	$\Pr \left( {\widehat {{C_{q|P}}} \ge {C_{q|P}}} \right) \ge 1 - \delta .$
\end{theorem}
\begin{proof}
	Observe Lemma~\ref{lemma:cp2d} and notice that in deterministic~settings, as shown in~\cite{HHHMitzenmacher}, $$\widehat f_q^ +  - \sum\limits_{h \in G\left( {q|P} \right)} {\widehat f_h^ - }  + \sum\limits_{h,h' \in G\left( {q|P} \right)} \widehat f_{{\rm{glb}}\left( {h,h'} \right)}^{+}$$\normalsize
	is a conservative estimate for ${C_{q|P}}$.
	Therefore, we need to account for the randomization error and verify that with probability $1-\delta$ it is less than $2{Z_{1 - \frac{\delta }{8}}}\sqrt {NV}$.
		
	We denote by $K$ the packets that may affect $C_{q|P}$.
	Since the expression of $\widehat{C_{q|P}}$ is not monotonic, we split it into two sets: $K^{+}$ are packets that affect  $\widehat{C_{q|P}}$ positively and  $K^{-}$ affect it negatively.
	Similarly, we define $\{Y_i^K\}$ to be Poisson random variables that represent how many of the packets of $K$ are in each bin.
	
	 We do not know how many bins affect the sum, but we know for sure that there are no more than $N$ balls.
	 We define the random variable $Y^K_+$ that defines the number of packets from $K$ that fell in the corresponding bins to have a positive impact on $\widehat{C_{q|P}}$.
	 Invoking Lemma~\ref{lemma:poissonConfidence} on $Y^K_+$ yields that:
$$\Pr \left( \left| {Y_K^ +  - E\left( {Y_K^ + } \right)} \right| \ge {Z_{1-\frac{\delta }{8}}}\sqrt \frac{N}{V} \right) \le \frac{\delta }{4}.$$
	Similarly, we define $Y_{K}^-$ to be the number of packets from $K$ that fell into the corresponding buckets to create a negative impact on $\widehat{C_{q|P}}$ and Lemma~\ref{lemma:poissonConfidence} results in:
$$\Pr \left( \left| {Y_K^ -  - E\left( {{Y_K}^ - } \right)} \right| \ge {Z_{1-\frac{\delta }{8}}}\sqrt \frac{N}{V } \right) \le \frac{\delta }{4}.$$
$Y_{K}^+$ is monotonically increasing with the number of balls and $Y_{K}^-$ is monotonically decreasing with the number of balls.
We can apply Lemma~\ref{lemma:rare} and conclude that:
\small
\[\begin{array}{l}
\Pr \left( {\widehat {{C_{q|P}}} \ge {C_{q|P}}} \right) \le \\
2\Pr \left( {V\left( {Y_K^ -  + Y_K^ + } \right) \ge \left( {VE\left( {Y_K^ -  + Y_K^ + } \right) + 2{Z_{1 - \frac{\delta }{8}}}\sqrt {NV} } \right)} \right)  \\
\le 1 - 2\frac{\delta }{2} = 1 - \delta ,
\end{array}\]\normalsize completing the proof.
\end{proof}

\subsubsection{Putting It All Together}
We can now prove the coverage property for one and two dimensions.
\begin{corollary}
	\label{cor:coverage}
	If $N>\NB$ then RHHH satisfies coverage.
	That is, given a prefix $q \notin P$, where $P$ is the set of HHH returned by RHHH,  $$\Pr \left(C_{q|P}<\theta N\right) >1-\delta.$$
\end{corollary}
\begin{proof}
	The proof follows form Theorem~\ref{thm:underCP} in one dimension, or
	Theorem~\ref{thm:conservativeCP} in two, that guarantee that in both cases: $\Pr \left( C_{q|P}<\widehat{C_{q|P}}\right) > 1-\delta$.
	
	The only case where $q \notin P$ is if $\widehat{C_{q|P}}<\theta N$.
    Otherwise, Algorithm~\ref{alg:Skipper} would have added it to $P$.
	However, with probability $1-\delta$, $C_{q|P}<\widehat{C_{q|P}}$, and therefore $C_{q|P} <\theta N$ as well.
\end{proof}

\subsection{RHHH Properties Analysis}
\label{sec:RHHH-prop}
Finally, we can prove the main result of our analysis. It establishes that if the number of packets is large enough, $RHHH$ is correct.
\begin{theorem}
	\label{thm:correctness}
	If $N>\NB$, then RHHH solves {\sc {$(\delta,\epsilon, \theta)$ - Approximate Hierarchical Heavy Hitters}}.
\end{theorem}
\begin{proof}
	The theorem is proved by combining\\ Lemma~\ref{lemma:accuracy} and Corollary~\ref{cor:coverage}.
\end{proof}

Note that $\NB \triangleq \NBound$ contains the parameter $V$ in it. 
When the minimal measurement interval is known in advance, the parameter $V$ can be set to satisfy correctness at the end of the measurement. 
For short measurements, we may need to use $V=H$, while longer measurements justify using $V\gg H$ and achieve better performance. 
When considering modern line speed and emerging new transmission technologies, this speedup capability is crucial because faster lines deliver more packets in a given amount of time and thus justify a larger value of $V$ for the same measurement~interval.

For completeness, we prove the following.
\begin{theorem}
	\label{thm:O1}
	RHHH's update complexity is $O(1)$.
\end{theorem}
\begin{proof}
Observe Algorithm~\ref{alg:Skipper}.
For each update, we randomize a number between $0$ and $V-1$, which can be done in $O(1)$.
Then, if the number is smaller than $H$, we also update a Space Saving instance, which can be done in $O(1)$ as well~\cite{SpaceSavings}.
\end{proof}
Finally, we note that our space requirement is similar to that of~\cite{HHHMitzenmacher}.
\begin{theorem}
	\label{thm:space}
	The space complexity of RHHH is $O\left(\frac{H}{\varepsilon_a}\right)$ flow table entries.
\end{theorem}
\begin{proof}
RHHH utilizes $H$ separate instances of Space Saving, each using $\frac{1}{\epsilon_a}$ table entries.
There are no other space significant data~structures.
\end{proof}

\section{Discussion}
\label{sec:discussion}
This work is about realizing hierarchical heavy hitters measurement in virtual network devices.
Existing HHH algorithms are too slow to cope with current improvements in network technology.
Therefore, we define a probabilistic relaxation of the problem and introduce a matching randomized algorithm called RHHH.
Our algorithm leverages the massive traffic in modern networks to perform simpler update operations.
Intuitively, the algorithm replaces the traditional approach of computing all prefixes for each incoming packets by sampling (if $V>H$) and then choosing one \emph{random} prefix to be updated.
While similar convergence guarantees can be derived for the simpler approach of updating all prefixes for each sampled packet, our solution has the clear advantage of processing elements in $O(1)$ worst case time.

We evaluated RHHH on four real Internet packet traces, consisting over 1 billion packets each and achieved a speedup of up to X62 compared to previous works.
Additionally, we showed that the solution quality of RHHH is comparable to that of previous work.
RHHH performs updates in constant time, an asymptotic improvement from previous works whose complexity is proportional to the hierarchy's size.
This is especially important in the two dimensional case as well as for IPv6 traffic that requires larger~hierarchies.

Finally, we integrated RHHH into a DPDK enabled Open vSwitch and evaluated its performance as well as the alternative algorithms.
We provided a dataplane implementation where HHH measurement is performed as part of the per packet routing tasks.
In a dataplane implementation, RHHH is capable of handling up to 13.8 Mpps, $4\%$ less than an unmodified DPDK OVS (that does not perform HHH measurement).
We showed a throughput improvement of X2.5 compared to the fastest dataplane implementations of previous~works.

Alternatively, we evaluated a distributed implementation where RHHH is realized in a virtual machine that can be deployed in the cloud and the virtual switch only sends the sampled traffic to RHHH.
Our distributed implementation can process up to 12.3 Mpps.
It is less intrusive to the switch, and offers greater flexibility in virtual machine placement.
Most importantly, our distributed implementation is capable of analyzing data from multiple network~devices.

Notice the performance improvement gap between our direct implementation -- X62, compared to the performance improvement when running over OVS -- X2.5.
In the case of the OVS experiments, we were running over a $10$Gbps link, and were bound by that line speed -- the throughput obtained by our implementation was only $4\%$ lower than the unmodified OVS baseline (that does nothing).
In contrast, previous works were clearly bounded by their computational overhead.
Thus, one can anticipate that once we deploy the OVS implementation on faster links, or in a setting that combines traffic from multiple links, the performance boost compared to previous work will be closer to the improvement we obtained in the direct~implementation.

A downside of RHHH is that it requires some minimal number of packets in order to converge to the desired formal accuracy guarantees.
In practice, this is a minor limitation as busy links deliver many millions of packets every second. 
For example, in the settings reported in Section~\ref{sec:acc+cov-error}, RHHH requires up to $100$ millions packets to fully converge, yet even after as little as $8$ millions packets, the error reduces to around $1\%$.
With a modern switch that can serve $10$ million packets per second, this translates into a $10$ seconds delay for complete convergence and around $1\%$ error after $1$ second.
As line rates will continue to improve, these delays would become even shorter accordingly. The code used in this work is open sourced~\cite{RHHHCode}

\paragraph*{Acknowledgments}
We thank Ori Rottenstreich for his insightful comments and Ohad Eytan for helping with the code release.
We would also like to thank the anonymous reviewers and our shepherd, Michael Mitzenmacher, for helping us improve this work.

This work was partially funded by the Israeli Science Foundation grant \#$1505/16$ and the Technion-HPI research school.
Marcelo Caggiani Luizelli is supported by the research fellowship program funded by CNPq (201798/2015-8).

{
	\bibliographystyle{plain}
	\bibliography{refs}
}
\end{document}